\documentclass[report,10pt,a4paper]{amsart}
\numberwithin{equation}{section}
\allowdisplaybreaks 

\usepackage{amssymb,eucal,mathrsfs,setspace}
\usepackage[noadjust]{cite}

\usepackage{array}
\newcolumntype{C}{>{$}c<{$}} 

\usepackage{mathptmx}                             
\DeclareSymbolFont{largesymbols}{OMX}{zplm}{m}{n} 

\usepackage[margin=26mm, head=10mm, foot=10mm]{geometry}

\usepackage{enumitem}
\setitemize{leftmargin=*}   
\setenumerate{leftmargin=*, 
  label=\textup{(\roman*)}} 

\usepackage{mathtools} 

\usepackage{graphicx}

\usepackage{tikz}
\usetikzlibrary{decorations.markings,positioning,arrows}
\tikzstyle{wt}=[circle, fill=black, inner sep=2pt, outer sep=0pt, minimum size=5pt]  
\tikzstyle{owt}=[circle, draw=black, fill=white, inner sep=2pt, outer sep=0pt, minimum size=5pt] 

\usepackage[colorlinks=true,citecolor=red,linkcolor=blue]{hyperref} 

\usepackage[capitalise,noabbrev]{cleveref} 

\newcommand{\pd}{\partial}         

\renewcommand{\ge}{\geqslant} 
\renewcommand{\le}{\leqslant} 

\DeclareMathOperator{\ad}{ad}

\DeclareMathOperator{\multiplicity}{mult}
\DeclareMathOperator{\sgn}{sgn}
\DeclareMathOperator{\vspn}{span}

\newcommand{\partitions}[1]{\mathcal{P}_{#1}}     
\newcommand{\specparts}[3]{\mathcal{P}(#1,#2;#3)}
\newcommand{\mult}[2]{\multiplicity_{#1}(#2)}     

\newcommand{\ii}{\sqrt{-1}} 
\newcommand{\qq}{\mathsf{q}}   
\newcommand{\yy}{\mathsf{y}}   
\newcommand{\zz}{\mathsf{z}}   

\newcommand{\wun}{\mathbf{1}}

\DeclarePairedDelimiter{\brac}{\lparen}{\rparen} 
\DeclarePairedDelimiter{\sqbrac}{\lbrack}{\rbrack} 
\DeclarePairedDelimiter{\set}{\lbrace}{\rbrace}
\newcommand{\st}{\mspace{5mu} {:} \mspace{5mu}} 
\DeclarePairedDelimiter{\abs}{\lvert}{\rvert}

\DeclarePairedDelimiter{\ang}{\langle}{\rangle}
\DeclarePairedDelimiter{\normord}{{:}}{{:}}        

\DeclarePairedDelimiterX{\comm}[2]{\lbrack}{\rbrack}{#1 , #2}  
\DeclarePairedDelimiterX{\acomm}[2]{\lbrace}{\rbrace}{#1 , #2} 
\DeclarePairedDelimiterX{\inner}[2]{\langle}{\rangle}{#1 , #2} 
\DeclarePairedDelimiterX{\super}[2]{\lparen}{\rparen}{#1 \delimsize\vert \mathopen{} #2} 

\DeclarePairedDelimiter{\ket}{\lvert}{\rangle}
\DeclarePairedDelimiterX{\braket}[2]{\langle}{\rangle}{#1 \delimsize\vert \mathopen{} #2}
\DeclarePairedDelimiterX{\bracket}[3]{\langle}{\rangle}{#1 \delimsize\vert \mathopen{} #2 \delimsize\vert \mathopen{} #3}

\newcommand{\lra}{\longrightarrow}
\newcommand{\ira}{\hookrightarrow}    
\newcommand{\lira}{\ensuremath{\lhook\joinrel\relbar\joinrel\rightarrow}} 

\newcommand{\dses}[5]{0 \lra #1 \overset{#2}{\lra} #3 \overset{#4}{\lra} #5 \lra 0} 

\newcommand{\fld}[1]{\mathbb{#1}}    
\newcommand{\alg}[1]{\mathfrak{#1}}  
\newcommand{\grp}[1]{\mathsf{#1}}    
\newcommand{\Mod}[1]{\mathcal{#1}}   
\newcommand{\VOA}[1]{\mathsf{#1}}    

\newcommand{\ZZ}{\fld{Z}}
\newcommand{\NN}{\ZZ_{\ge 0}} 

\newcommand{\RR}{\fld{R}}
\newcommand{\CC}{\fld{C}}

\newcommand{\SLG}[2]{\grp{#1}_{#2}}       

\newcommand{\SLA}[2]{\alg{#1}_{#2}}                       
\newcommand{\SLSA}[3]{\alg{#1} \super{#2}{#3}}            
\newcommand{\envalg}[1]{\mathsf{U}\brac{#1}}              

\newcommand{\tideal}[1]{\ang{#1}}

\newcommand{\kil}[2]{\kappa \brac{#1, #2}} 

\newcommand{\finite}[1]{\overline{#1}}
\newcommand{\affine}[1]{#1}
\newcommand{\hfour}{\finite{\alg{h}}_4}    
\newcommand{\ahfour}{\affine{\alg{h}}_4}   

\newcommand{\confwtsymb}{\Delta}                   
\newcommand{\hwconfwt}[1]{\confwtsymb^+_{#1}}      
\newcommand{\lwconfwt}[1]{\confwtsymb^-_{#1}}      
\newcommand{\hlwconfwt}[1]{\confwtsymb^{\pm}_{#1}} 
\newcommand{\rhwconfwt}[2]{\confwtsymb_{#1;#2}}    

\newcommand{\autfont}[1]{\mathrm{#1}}
\newcommand{\fconjaut}{\finite{\autfont{c}}}         
\newcommand{\fscaut}{\finite{\autfont{r}}}           
\newcommand{\fshaut}{\finite{\autfont{s}}}           
\newcommand{\conjaut}{\autfont{c}}                   
\newcommand{\scaut}{\autfont{r}}                     
\newcommand{\shaut}{\autfont{s}}                     
\newcommand{\sfaut}{\sigma}                          

\newcommand{\func}[2]{#1(#2)}                        
\newcommand{\fconjmod}[1]{\func{\fconjaut}{#1}}      
\newcommand{\fscmod}[2]{\func{\fscaut_{#1}}{#2}}     
\newcommand{\fshmod}[2]{\func{\fshaut_{#1}}{#2}}     
\newcommand{\conjmod}[1]{\func{\conjaut}{#1}}        
\newcommand{\shmod}[2]{\func{\shaut'_{#1}}{#2}}      
\newcommand{\sfmod}[2]{\func{\sfaut^{#1}}{#2}}       

\newcommand{\Hfour}{\VOA{H}_4}                          

\newcommand{\fmod}[1]{\finite{\Mod{#1}}}
\newcommand{\amod}[1]{\affine{\Mod{#1}}}

\newcommand{\VerSymb}{V}
\newcommand{\IrrSymb}{L}
\newcommand{\RelSymb}{R}
\newcommand{\MaxSubSymb}{M}
\newcommand{\SubSymb}{N}
\newcommand{\QuoSymb}{E}
\newcommand{\QuoCohSymb}{C}
\newcommand{\RelCohSymb}{D}

\newcommand{\finver}[1]{\fmod{\VerSymb}_{#1}}           
\newcommand{\finirr}[1]{\fmod{\IrrSymb}_{#1}}           
\newcommand{\finrel}[2]{\fmod{\RelSymb}_{#1; #2}}       

\newcommand{\ver}[1]{\amod{\VerSymb}_{#1}}              
\newcommand{\irr}[1]{\amod{\IrrSymb}_{#1}}              
\newcommand{\rel}[2]{\amod{\RelSymb}_{#1; #2}}          
\newcommand{\maxsub}[2]{\amod{\MaxSubSymb}_{#1; #2}}    
\newcommand{\sub}[2]{\amod{\SubSymb}_{#1; #2}}          
\newcommand{\quo}[2]{\amod{\QuoSymb}_{#1; #2}}          
\newcommand{\qcoh}[2]{\amod{\QuoCohSymb}_{#1; #2}}      
\newcommand{\rcoh}[2]{\amod{\RelCohSymb}_{#1; #2}}      

\DeclareMathOperator{\tr}{tr}

\newcommand{\Gr}[1]{\sqbrac[\big]{#1}}                 
\newcommand{\tGr}[1]{\sqbrac{#1}}                      

\newcommand{\traceover}[1]{\tr_{\raisebox{-2pt}{$\scriptstyle #1$}}} 

\DeclareMathOperator{\chmap}{ch}
\DeclareMathOperator{\strfunsymb}{s}

\newcommand{\ch}[1]{\chmap \Gr{#1}}                    
\newcommand{\fch}[3]{\ch{#1} \brac[\big]{#2;#3}}       

\newcommand{\strfun}[2]{\strfunsymb_{#1} \Gr{#2}}      
\newcommand{\fstrfun}[3]{\strfun{#1}{#2} \brac[\big]{#3}}
\newcommand{\tstrfun}[2]{\strfunsymb_{#1} \tGr{#2}}

\newcommand{\jth}[1]{\vartheta_{#1}}                   
\newcommand{\fjth}[2]{\jth{#1} \brac{#2}}              

\newcommand{\cft}{conformal field theory}
\newcommand{\cfts}{conformal field theories}
\newcommand{\lcft}{logarithmic conformal field theory}

\newcommand{\uea}{universal enveloping algebra}

\newcommand{\lw}{lowest-weight}
\newcommand{\lwv}{\lw\ vector}
\newcommand{\lwvs}{\lwv s}
\newcommand{\lwm}{\lw\ module}

\newcommand{\hw}{highest-weight}
\newcommand{\hwv}{\hw\ vector}
\newcommand{\hwvs}{\hwv s}
\newcommand{\hwm}{\hw\ module}
\newcommand{\hwms}{\hwm s}
\newcommand{\rhw}{relaxed highest-weight}
\newcommand{\rhwv}{\rhw\ vector}
\newcommand{\rhwvs}{\rhwv s}
\newcommand{\rhwm}{\rhw\ module}
\newcommand{\rhwms}{\rhwm s}
\newcommand{\sv}{singular vector}
\newcommand{\svs}{\sv s}
\newcommand{\vo}{vertex operator}
\newcommand{\voa}{\vo\ algebra}
\newcommand{\voas}{\voa s}

\newcommand{\ope}{operator product expansion}
\newcommand{\opes}{\ope s}

\newcommand{\fdim}{finite-dimensional}

\newcommand{\km}{Kac--Moody}
\newcommand{\nw}{Nappi--Witten}

\newcommand{\pbw}{Poincar\'{e}--Birkhoff--Witt}
\newcommand{\wzw}{Wess--Zumino--Witten}

\theoremstyle{plain}
\newtheorem{theorem}{Theorem}
\newtheorem{corollary}[theorem]{Corollary}
\newtheorem{lemma}[theorem]{Lemma}
\newtheorem{proposition}[theorem]{Proposition}

\makeatletter
\renewcommand\author@andify{%
  \nxandlist {\unskip ,\penalty-1 \space\ignorespaces}%
    {\unskip {} \@@and~}%
    {\unskip \penalty-2 \space \@@and~}%
}
\makeatother

\begin{document}

\title{Representations of the Nappi--Witten vertex operator algebra}

\author[A~Babichenko]{Andrei Babichenko}
\address[Andrei Babichenko]{
Department of Mathematics \\
Weizmann Institute of Science \\
Rehovot, 76100, Israel.
}
\email{babichenkoandrei@gmail.com}

\author[K~Kawasetsu]{Kazuya Kawasetsu}
\address[Kazuya Kawasetsu]{
Priority Organization for Innovation and Excellence\\
Kumamoto University\\
Kumamoto, Japan, 860-8555.
}
\email{kawasetsu@kumamoto-u.ac.jp}

\author[D~Ridout]{David Ridout}
\address[David Ridout]{
School of Mathematics and Statistics \\
University of Melbourne \\
Parkville, Australia, 3010.
}
\email{david.ridout@unimelb.edu.au}

\author[W~Stewart]{William Stewart}
\address[William Stewart]{
Department of Mathematics \\
University of Texas at Austin \\
Austin, USA, 78712.
}
\email{wbstewart@utexas.edu}


\begin{abstract}
	The \nw\ model is a \wzw\ model in which the target space is the nonreductive Heisenberg group $H_4$.  We consider the representation theory underlying this \cft.  Specifically, we study the category of weight modules, with \fdim\ weight spaces, over the associated affine \voa\ $\Hfour$.  In particular, we classify the irreducible $\Hfour$-modules in this category and compute their characters.  We moreover observe that this category is nonsemisimple, suggesting that the \nw\ model is a \lcft.
\end{abstract}

\maketitle

\markleft{A~BABICHENKO, K~KAWASETSU, D~RIDOUT AND W~STEWART} 

\onehalfspacing

\section{Introduction}

\wzw\ models are examples of nonlinear sigma models that describe (noncritical) strings propagating on Lie groups or supergroups \cite{WitNon84}.  The enhanced symmetry afforded by Lie-theoretic target spaces makes them attractive studies and consequently the models based on reductive Lie groups are now very well understood.  In particular, \wzw\ models with reductive Lie groups have chiral symmetry algebras that may be identified as affine \voas\ at nonnegative integer levels.  This identification not only algebraically formalises the conformal symmetry of these models, it also means that the representation theory of the corresponding affine \km\ algebras \cite{KacInf90} is available to organise the spectrum.

In contrast, the \wzw\ models with nonreductive Lie groups are relatively poorly understood.  It has been known for quite some time \cite{FigNon94,FigNon96} that the theory is conformal if the corresponding Lie algebra admits a nondegenerate invariant symmetric bilinear form.  However, the representation theory of the corresponding affine \voas\ has remained largely unexplored, probably because the representation theory of nonreductive Lie algebras is already considerably more challenging than the reductive case.  In this paper, we revisit one of the simplest nonreductive examples: the \nw\ model \cite{NapWZW93} corresponding to the Heisenberg group $H_4$.

The \nw\ model was originally introduced to describe strings propagating in a monochromatic plane-wave background, a seemingly simple generalisation of the flat backgrounds familiar from abelian group target spaces, and has since been intensively studied.  For example, this model lends itself to the study of the propagation and scattering of strings in the presence of gravitational waves \cite{KirStr94,DApStr03}.  Moreover, its rich collection of abelian and diagonal cosets have been related to many other interesting backgrounds \cite{DApAbe07,DApDia08}.  Boundary aspects (D-branes) have likewise been explored in detail, see \cite{StaDBr98,FigMor00,DApDBr05} for example.

In a sense, the representation theory of the \nw\ \voa\ $\Hfour$ has been discussed many times in the literature.  In particular, a spectrum has been proposed from which $3$- and $4$-point correlation functions have been calculated, see \cite{DApStr03} for example.  However, this proposal largely follows the paradigm of noncompact \cft\ introduced in \cite{MalStr01} for the $\SLG{SL}{2}(\RR)$ \wzw\ model: start with the unitary representations of the \nw\ Lie algebra $\hfour$ and induce to get representations of its affinisation $\ahfour$.  It therefore identifies the generic features of the spectrum, but may miss the finer structure needed for a complete understanding.

Our primary motivation here is that an alternative paradigm has recently arisen through the study of \wzw\ models on Lie supergroups \cite{SalGL106,CreRel11,QueSup13}, fractional-level \wzw\ models \cite{GabFus01,LesLog04,CreLog13,RidAdm17} and related theories such as the bosonic ghost system \cite{LesLog04,RidBos14,AdaFus19,AllBos20}.  As \cfts, these share similar features to the known noncompact models such as continuous spectra and conformal dimensions that are unbounded from below.  More interestingly, the fine structure of the spectra of these theories shows that they are \emph{logarithmic}, meaning that the spectrum receives contributions from reducible but indecomposable representations with a nondiagonalisable action of the hamiltonian \cite{GurLog93}.  Moreover, these logarithmic representations are necessary for the standard consistency checks of closure under fusion \cite{GabFus01,RidFus10} and modular invariance \cite{CreMod13,RidBos14}.

It would therefore be very interesting to study the fine structure of the spectrum of the \nw\ model.  In this respect, we note that \cite{DApStr03}, see also \cite{DApAbe07}, find strong evidence for logarithmic singularities in the $4$-point correlators, usually regarded as a sure sign of logarithmic representations.  However, they argue that the \nw\ model is not logarithmic and that these singularities arise somehow because of the continuous nature of the spectrum.  On the other hand, there are related coset models that are known to be logarithmic \cite{BakPPW02,SfeStr02}.

A first step towards understanding the fine structure of the \nw\ model was taken in \cite{BaoRep11}, where the \hwms\ of the affinisation $\ahfour$ of $\hfour$ were studied.  The results included a precise determination of when a Verma $\ahfour$-module is irreducible and a partial characterisation of the submodule structure when it is not.  They also conjectured similar results for certain generalised Verma modules, here referred to as \emph{relaxed} Verma modules following \cite{FeiEqu98,RidRel15}.  Our main aim in this paper is to complete these results.  We reprove (in a simpler fashion) their classification result and moreover completely determine the substructures of all Verma modules.  We also prove their conjectures for relaxed Verma modules.  Consequently, we deduce character formulae for all irreducible weight representations, with \fdim\ weight spaces, over the \nw\ \voa\ $\Hfour$.  We also prove that $\Hfour$ admits many indecomposable but reducible representations, though those that we construct are not themselves logarithmic.

Two obvious questions, which we leave for future work, are to ascertain the modular transformations of the irreducible characters and to determine the fusion rules.  Both are necessary checks on the consistency on any proposed spectrum.  To check modularity, we expect to apply the \emph{standard module formalism} of \cite{CreLog13,RidVer14}, though there are some obvious technical considerations to overcome.  Computing fusion rules is expected to be much more difficult, but here there are already some helpful partial results on tensor products of $\hfour$-modules \cite{VilSpe68,MilLie68}.  One main aim here would be to determine, given a proposed spectrum, if logarithmic representations are generated as fusion products of the already known nonlogarithmic representations.

This paper is organised as follows.  We introduce our notation and conventions for the Heisenberg Lie algebra $\hfour$ in \cref{sec:hfour}.  We also discuss its automorphisms and classify its irreducible weight representations.  The latter result is very well known, see \cite{MilLie68} for example, but we outline the easy proof for completeness and because it is convenient for constructing certain, less well known, families of reducible but indecomposable $\hfour$-representations.  This is followed, in \cref{sec:Hfour}, with a summary of our notation and conventions for the affinisation $\ahfour$ and the associated universal affine \voa\ $\Hfour$.  Here, we also discuss automorphisms including those of ``spectral flow'' type.

\cref{sec:hwH4} is devoted to a thorough study of the \hwms\ of $\Hfour$.  In particular, we give an elementary proof of the fact that $\Hfour$ is a simple \voa\ and we rigorously identify the maximal submodule of an arbitrary Verma $\ahfour$-module (\cref{thm:bjp}).  We give a combinatorial proof of the latter result in \cref{app:proof}.  The generalisation to \rhwms\ is the subject of \cref{sec:rhwH4}.  Here, we adapt the methodology developed in \cite{KawRel18} to study \rhwms\ over the admissible-level affine \vo\ (super)algebras associated to $\SLA{sl}{2}$ and $\SLSA{osp}{1}{2}$.  Our main result, \cref{prop:classH-mod}, classifies all the irreducible \rhwms\ of $\Hfour$.  We also describe the structure of certain reducible but indecomposable $\Hfour$-modules in \cref{th:IndecompSES}.

Our last task is to compute the characters of the irreducible \rhw\ $\Hfour$-modules.  In \cref{sec:CharsI}, we give the characters of the Verma and relaxed Verma modules, which are straightforward to derive, and those of the irreducible \hwms, which follow from \cref{thm:bjp}.  \cref{sec:Stringiness} is devoted to the much more subtle computation of the characters of the irreducible \rhwms.  This again uses the methodology developed in \cite{KawRel18} and moreover proves the main conjecture of \cite{BaoRep11}.  We conclude with a short discussion concerning directions for future work.

We remark that the methodology of \cite{KawRel18} is but one of many recent approaches being developed to explore the theory of \rhwms\ for affine \voas\ (and their associated W-algebras), see \cite{AraWei16,AdaRea17,KawRel19,FutPos20,FutSim20,KawRel20,AdaRea20,FehCla20,FutAdm21} for example.  As mentioned above, the main novelty of analysing these modules over $\Hfour$ is, in our opinion, that this example derives from a nonreductive Lie algebra.\footnote{In a sense, $\Hfour$ is the bosonic analogue of the well-studied $\SLSA{gl}{1}{1}$ \vo\ superalgebra \cite{RozQua92,SalGL106,CreRel11,CreWAl11,CreTen20b}.  We find it interesting that the latter has a far more accessible representation theory than the former.}  It would be extremely interesting to see how these other approaches can also accommodate such nonreductive cases.

\section*{Acknowledgements}

We are thankful to Cuipo Jiang and Thomas Quella for interesting discussions related to this research.
KK's research is partially supported by MEXT Japan ``Leading Initiative for Excellent Young Researchers (LEADER)'',
JSPS Kakenhi Grant numbers 19KK0065 and 19J01093 and
Australian Research Council Discovery Project DP160101520.
DR's research is supported by the Australian Research Council Discovery Project DP160101520 and the Australian Research Council Future Fellowship FT200100431.
WS's research is supported by an Australian Government Research Training Program (RTP) Scholarship.

\section{The Lie algebra $\hfour$ and its representations} \label{sec:hfour}

The Lie algebra $\hfour$ is the four-dimensional complex Lie algebra with basis $\set*{E,F,I,J}$ whose nonzero Lie brackets are, modulo antisymmetry, as follows:
\begin{equation}
	\comm{E}{F} = I, \quad \comm{J}{E} = E, \quad \comm{J}{F} = -F.
\end{equation}
As $I$ spans the centre, but the centre has no complementary ideal, $\hfour$ is not reductive.  On the other hand, the Killing form is easily checked to be nonzero, so $\hfour$ is not solvable.  Nevertheless, it admits a two-parameter family of nondegenerate invariant symmetric bilinear forms, given by
\begin{equation} \label{eq:DefKil}
	\kil{E}{F} = \kil{I}{J} = a, \quad \kil{J}{J} = b, \qquad a \in \CC \setminus \set{0},\ b \in \CC,
\end{equation}
with all other entries $0$.  For reasons that will shortly become clear, we may take $a=1$ and $b=0$ in what follows.

This Lie algebra also possesses a natural triangular decomposition:
\begin{equation} \label{eq:FinTriDec}
	\hfour = \hfour^+ \oplus \hfour^0 \oplus \hfour^-; \qquad
	\hfour^+ = \vspn \set{E}, \quad \hfour^0 = \vspn \set{I,J}, \quad \hfour^- = \vspn \set{F}.
\end{equation}
A \hwv{} is then a simultaneous eigenvector of $I$ and $J$ which is annihilated by $E$ and a \lwv{} is the same except that the annihilation is by $F$.  The span of a given \hwv{} is then naturally a $(\hfour^+ \oplus \hfour^0)$-module and inducing to a $\hfour$-module defines the corresponding Verma module.  If the $I$- and $J$-eigenvalues of the \hwv{} are $i$ and $j$, respectively, then we shall denote the induced (\hw{}) Verma module by $\finver{i,j}^+$.

Let $\ket{i,j}$ denote the \hwv{} of $\finver{i,j}^+$.  Then, a basis for $\finver{i,j}^+$ is given by the $F^n \ket{i,j}$, where $n \in \NN$.  If $\finver{i,j}^+$ has a nonzero proper submodule, then one of the $F^n \ket{i,j}$ (with $n>0$) must be a \sv{}.  But, $\comm{E}{F} = I$ is central, hence
\begin{equation}
	EF^n \ket{i,j} = n F^{n-1} I \ket{i,j} = in F^{n-1} \ket{i,j}
\end{equation}
and we conclude that such \svs{} only exist if $i=0$.  Moreover, if $i=0$, then all of the $F^n \ket{i,j}$ with $n>0$ are singular.  It follows that the Verma module $\finver{i,j}^+$ of $\hfour$ is irreducible if and only if $i \ne 0$ and that the maximal submodule of $\finver{0,j}^+$ is isomorphic to $\finver{0,j-1}^+$.  The Verma module $\finver{i,j}^+$ is then irreducible if $i \ne 0$ and, for $i=0$, we instead have the short exact sequence
\begin{equation} \label{es:finVVL}
	\dses{\finver{0,j-1}^+}{}{\finver{0,j}^+}{}{\finirr{0,j}}.
\end{equation}
Note that the $\finirr{0,j}$ are all one-dimensional and are therefore irreducible.  We therefore have a complete classification of irreducible \hwms{} for $\hfour$.
\noindent We record the following consequence of this classification that will be used later.
\begin{lemma} \label{le:IZero}
	If $V$ is a finite-dimensional weight $\hfour$-module, then $I$ acts trivially on $V$.
\end{lemma}
\begin{proof}
	The composition factors of $V$ are obviously finite-dimensional, hence each is isomorphic to one of the $\finirr{0,j}$.  It follows that $I$ acts as zero on each factor.  However, $V$ is weight, so $I$ acts as zero on $V$ as well.
\end{proof}
\noindent Alternatively, $I$ acts as a constant multiple of the identity on each composition factor, by Schur's lemma.  But, taking the trace of $I=\comm{E}{F}$ over this finite-dimensional space shows that this multiple must be $0$.

One can similarly analyse \lw{} Verma modules $\finver{i,j}^-$.  However, their structure follows immediately from the existence of an automorphism $\fconjaut$ of $\hfour$ defined by
\begin{equation} \label{eq:DefFinConj}
	\fconjaut(E) = -F, \quad \fconjaut(I) = -I, \quad \fconjaut(J) = -J, \quad \fconjaut(F) = -E.
\end{equation}
We call the $\fconjaut$ the conjugation automorphism and note that it squares to the identity.  This lifts to an invertible endofunctor on the category of weight modules for $\hfour$ as follows.  Given such a module $\fmod{W}$, let $\func{\fconjaut^*}{\fmod{W}}$ denote the image of $\fmod{W}$ under an arbitrary vector space isomorphism $\fconjaut^*$.  We equip $\func{\fconjaut^*}{\fmod{W}}$ with an $\hfour$-module structure by defining
\begin{equation} \label{eq:DefTwMod}
	A \cdot \fconjaut^*(m) = \fconjaut^*(\fconjaut^{-1}(A) m), \qquad A \in \hfour,\ m \in \fmod{W}.
\end{equation}
In other words, $\fconjaut(A) \fconjaut^*(m) = \fconjaut^*(Am)$.  In what follows, we shall, for brevity, drop the star that distinguishes the automorphism from the associated category autoequivalence.

As $\fconjaut$ defines an invertible functor, conjugation preserves module structure.  In particular, it maps irreducibles to irreducibles.  For example, the conjugate of the irreducible \hw{} module $\finver{i,j}^+$, $i \ne 0$, is the irreducible \lwm{} $\finver{-i,-j}^-$ of \lw{} $(-i,-j)$ and the conjugate of the one-dimensional module $\finirr{0,j}$ is $\finirr{0,-j}$ (which is simultaneously highest- and \lw).  Moreover, invertibility means that these modules exhaust the irreducible \lw{} $\hfour$-modules, up to isomorphism.

There are other nontrivial automorphisms of $\hfour$.  In particular, we have two one-parameter families that we shall refer to as the rescale automorphisms $\fscaut_{\alpha}$, $\alpha \in \CC \setminus \set{0}$, and the shift automorphisms $\fshaut_{\beta}$, $\beta \in \CC$, defined by
\begin{equation} \label{eq:DefFinReShAut}
	\begin{aligned}
		\fscaut_{\alpha}(E) &= \alpha^{-1} E, & \fscaut_{\alpha}(I) &= \alpha^{-2} I, & \fscaut_{\alpha}(J) &= J, & \fscaut_{\alpha}(F) &= \alpha^{-1} F, \\
		\fshaut_{\beta}(E) &= E, & \fshaut_{\beta}(I) &= I, & \fshaut_{\beta}(J) &= J - \beta I, & \fshaut_{\beta}(F) &= F.
	\end{aligned}
\end{equation}
Along with conjugation, these automorphisms satisfy
\begin{equation}
	\fscaut_{\alpha} \fscaut_{\alpha'} = \fscaut_{\alpha \alpha'}, \quad
	\fshaut_{\beta} \fshaut_{\beta'} = \fshaut_{\beta + \beta'}, \quad
	\fscaut_{\alpha} \fconjaut = \fconjaut \fscaut_{\alpha}, \quad
	\fshaut_{\beta} \fconjaut = \fconjaut \fshaut_{\beta}, \quad
	\fscaut_{\alpha} \fshaut_{\beta} = \fshaut_{\alpha^{-2} \beta} \fscaut_{\alpha}.
\end{equation}
Note that the rescale and shift automorphisms do not preserve the choice of nondegenerate invariant bilinear form (unlike conjugation).  Instead, $\fscaut_{\alpha}$ has the effect of replacing the parameter $a$ in \eqref{eq:DefKil} by $\alpha^{-2} a$ and $\fshaut_{\beta}$ similarly replaces $b$ by $b - 2 \beta a$.  As $a \ne 0$, these automorphisms effectively allow us to tune $a$ and $b$ to any values we desire.  This shows that there was no loss of generality in choosing $a=1$ and $b=0$, as we did above.

As with conjugation, the rescale and shift automorphisms lift to invertible, and thus structure-preserving, endofunctors on the category of weight modules.  Since these automorphisms preserve the triangular decomposition \eqref{eq:FinTriDec} (again unlike conjugation), the corresponding endofunctors preserve being \hw{} or \lw{}.  An easy calculation with \hwvs{} now shows that
\begin{equation}
	\fscmod{\alpha}{\finver{i,j}^+} = \finver{\alpha^2 i,j}^+, \quad
	\fshmod{\beta}{\finver{i,j}^+} = \finver{i,j+\beta i}^+, \quad
	\fscmod{\alpha}{\finirr{0,j}} = \finirr{0,j}, \quad
	\fshmod{\beta}{\finirr{0,j}} = \finirr{0,j+\beta i}.
\end{equation}
In particular, this explains why the structure of the Verma modules of $\hfour$ is independent of the eigenvalue of $J$ and only depends on whether the eigenvalue of $I$ is zero or not.

To complete the classification of irreducible weight modules for $\hfour$, we consider modules without highest- or \lwvs{}.\footnote{Our definition of weight module will always require that the dimension of the weight spaces is finite.}  For this, it is convenient to introduce a central element of quadratic degree in the \uea{} $\envalg{\hfour}$:
\begin{equation} \label{eq:DefCas}
	Q = FE+IJ.
\end{equation}
It follows that the eigenvalues of $Q$ on $\finirr{0,j}$ and $\finver{i,j}^+$ are $0$ and $ij$, respectively.  This quadratic Casimir can, of course, be modified by adding an arbitrary polynomial in $I$ without affecting its central nature.

The reason for introducing $Q$ is that any weight space of an irreducible weight module for $\hfour$ defines an irreducible module for the centraliser of $\hfour^0$ in $\envalg{\hfour}$ and this centraliser is easily seen to be just the polynomial algebra $\CC[I, J, Q]$.  Because the centraliser $\CC[I, J, Q]$ is abelian, the weight spaces of any irreducible weight module of $\hfour$ are one-dimensional.  Moreover, given an irreducible $\CC[I, J, Q]$-module $\CC \ket{i,j;h}$, where $\ket{i,j;h}$ is labelled by its $I$-, $J$- and $Q$-eigenvalues (in that order), we can induce to a $\envalg{\hfour}$-module $\finrel{i,j}{h}$ and a basis for the latter is given by $\ket{i,j;h}$ and the $E^n \ket{i,j;h}$ and $F^n \ket{i,j;h}$ with $n \in \ZZ_{>0}$.  The $\finrel{i,j}{h}$ are said to be dense because their weight supports are maximal among those of all indecomposable $\hfour$-modules.

If $\finrel{i,j}{h}$ is reducible, then either one of the $E^n \ket{i,j;h}$ is a \lwv{} or one of the $F^n \ket{i,j;h}$ is a \hwv{}.  But, if $n$ is a positive integer for which $E^n \ket{i,j;h}$ is a \lwv{}, then
\begin{equation}
0 = F E^n \ket{i,j;h} = \brac*{Q-IJ} E^{n-1} \ket{i,j;h} = \brac[\big]{h-i(j+n-1)} E^{n-1} \ket{i,j;h}.
\end{equation}
Thus, $h=i(j+m)$ for some $m \in \NN$.  Similarly, if $n$ is a positive integer for which $F^n \ket{i,j;h}$ is a \hwv{}, then $h=i(j+m)$ for some $m \in \ZZ_{<0}$.  A necessary and sufficient condition for reducibility is thus that $h=i(j+m)$ for some $m \in \ZZ$.

Note that if the $\hfour$-module $\finrel{i,j}{h}$ is irreducible, then we have $\finrel{i,j}{h} \cong \finrel{i,j+n}{h}$, for any $n \in \ZZ$.  We shall therefore denote the irreducible dense modules by $\finrel{i,[j]}{h}$, where $[j] \in \CC / \ZZ$.  It will be convenient for what follows to pick a basis $\set*{\ket{i,j';h} \st j' \in [j]}$ of weight vectors of $\finrel{i,[j]}{h}$.  Here, $i$, $j'$ and $h$ are the eigenvalues of $I$, $J$ and $Q$, respectively.

We have thus arrived at the following classification of irreducible weight modules.
\begin{proposition} \label{prop:finclass}
	The irreducible weight modules of $\hfour$ are classified, up to isomorphism, by the following list of mutually inequivalent modules:
	\begin{itemize}
		\item The $\finirr{0,j}$ with $j \in \CC$.  These are one-dimensional and have both a highest- and \lwv{}.
		\item The $\finver{i,j}^+$ with $i \in \CC \setminus \set{0}$ and $j \in \CC$.  These are infinite-dimensional and have a \hwv{} but no \lwv{}.
		\item The $\finver{i,j}^-$ with $i \in \CC \setminus \set{0}$ and $j \in \CC$.  These are infinite-dimensional and have a \lwv{} but no \hwv{}.
		\item The $\finrel{i,[j]}{h}$ with $i,h \in \CC$, $[j] \in \CC / \ZZ$ and $h \notin i[j]$.  These are infinite-dimensional and have neither a \lwv{} nor a \hwv{}.
	\end{itemize}
\end{proposition}

It is also easy to analyse the reducible dense modules $\finrel{i,j}{h}$ that arise when $h \in i[j]$.  As above, for $m,n \in \ZZ_{>0}$, we know that $E^m \ket{i,j;h}$ is a \lwv{} if $h=i(j+m-1)$ and $F^n \ket{i,j;h}$ is a \hwv{} if $h=i(j-n)$.  It follows that $\finrel{i,j}{h}$ has both a \hw{} and a \lwv{} if and only if $h=i=0$, in which case it has infinitely many (every positive $m$ and $n$ works).  $\finrel{0,j}{0}$ thus has infinitely many composition factors $\finirr{0,j+\ell}$, one for each $\ell \in \ZZ$.  The indecomposable structure of $\finrel{0,j}{0}$ may be characterised through either of the following nonsplit short exact sequences:
\begin{equation} \label{es:finR}
	\begin{gathered}
		\dses{\finver{0,j-1}^+}{}{\finrel{0,j}{0}}{}{\finver{0,j}^-}, \\
		\dses{\finver{0,j-1}^+ \oplus \finver{0,j+1}^-}{}{\finrel{0,j}{0}}{}{\finirr{0,j}}, \\
		\dses{\finver{0,j+1}^-}{}{\finrel{0,j}{0}}{}{\finver{0,j}^+}.
	\end{gathered}
\end{equation}
On the other hand, $\finrel{0,j}{h} = \finrel{0,[j]}{h}$ is always irreducible for $h\ne0$, in accordance with \cref{prop:finclass}.

When $i\ne0$ and $h \in i[j]$, $\finrel{i,j}{h}$ has either a unique \hwv{} or a unique \lwv{}, corresponding to taking $m=\frac{h}{i}-j+1$ or $n=-\frac{h}{i}+j$, respectively.  We shall denote these reducible but indecomposable dense $\hfour$-modules by $\finrel{i}{h}^+$ and $\finrel{i}{h}^-$, respectively, because their $J_0$-eigenvalues $[j] = [h/i]$ are uniquely determined by $h$ and $i$.  They are characterised by the following nonsplit short exact sequences, valid for $i\ne0$:
\begin{equation} \label{es:finRpm}
	\begin{gathered}
		\dses{\finver{i,h/i}^+}{}{\finrel{i}{h}^+}{}{\finver{i,h/i+1}^-}, \\
		\dses{\finver{i,h/i+1}^-}{}{\finrel{i}{h}^-}{}{\finver{i,h/i}^+}.
	\end{gathered}
\end{equation}
There are of course many other reducible but indecomposable dense $\hfour$-modules, particularly when $i=0$.

It is easy to identify the result of conjugating, rescaling and shifting the dense irreducibles $\finrel{i,[j]}{h}$, $h \notin i[j]$:
\begin{equation}
	\fconjmod{\finrel{i,[j]}{h}} \cong \finrel{-i,[-j]}{h+i}, \quad
	\fscmod{\alpha}{\finrel{i,[j]}{h}} \cong \finrel{\alpha^2 i,[j]}{\alpha^2 h}, \quad
	\fshmod{\beta}{\finrel{i,[j]}{h}} \cong \finrel{i,[j+\beta i]}{h+\beta i^2}.
\end{equation}
The corresponding results for the reducible versions $\finrel{i}{h}^{\pm}$, $i\ne0$, and $\finrel{0,j}{0}$, $j\in\CC$, are
\begin{equation}
	\begin{aligned}
		\fconjmod{\finrel{i}{h}^{\pm}} &\cong \finrel{-i}{h+i}^{\mp}, &
		\fscmod{\alpha}{\finrel{i}{h}^{\pm}} &\cong \finrel{\alpha^2 i}{\alpha^2 h}^{\pm}, &
		\fshmod{\beta}{\finrel{i}{h}^{\pm}} &\cong \finrel{i}{h+\beta i^2}^{\pm}, \\
		\fconjmod{\finrel{0,j}{0}} &\cong \finrel{0,-j}{0}, &
		\fscmod{\alpha}{\finrel{0,j}{0}} &\cong \finrel{0,j}{0}, &
		\fshmod{\beta}{\finrel{0,j}{0}} &\cong \finrel{0,j}{0}.
	\end{aligned}
\end{equation}

\section{The affine algebra $\ahfour$ and the vertex algebra $\Hfour$} \label{sec:Hfour}

Because $\hfour$ has a nondegenerate invariant symmetric bilinear form \eqref{eq:DefKil}, it has a well defined affinisation
\begin{equation}
	\ahfour = \hfour \otimes \CC[t,t^{-1}] \oplus \CC K.
\end{equation}
Writing $A_n$ for $A \otimes t^n$, where $A \in \hfour$ and $n \in \ZZ$, the Lie brackets are given by
\begin{equation} \label{cr:ah4}
	\comm{A_m}{B_n} = \comm{A}{B}_{m+n} + m \kil{A}{B} \delta_{m+n,0} K, \quad \comm{A_m}{K} = 0, \qquad A,B \in \hfour,\ m,n \in \ZZ.
\end{equation}
It follows that $I_0$ is also central in $\ahfour$.  We also have a generalised triangular decomposition given by
\begin{equation} \label{triang}
	\begin{gathered}
		\ahfour = \ahfour^+ \oplus \ahfour^0 \oplus \ahfour^-;\\
		\ahfour^+ = \vspn \set{A_n \st A \in \hfour,\ n \in \ZZ_{>0}}, \quad
		\ahfour^0 = \vspn \set{A_0, K \st A \in \hfour}, \quad
		\ahfour^- = \vspn \set{A_n \st A \in \hfour,\ n \in \ZZ_{<0}}.
	\end{gathered}
\end{equation}
From this, we obtain a parabolic Verma module for $\ahfour$ by taking the one-dimensional representation $\finirr{0,0}$ of $\hfour$, extending it to an $\ahfour^+ \oplus \ahfour^0$-module by letting the $A_n$ with $n>0$ act as $0$ and $K$ act as multiplication by $k \in \CC$, and then inducing to a $\ahfour$-module.  The constant $k$ is called the level.

This parabolic Verma module carries the structure of a vertex algebra given by the standard affine state-field correspondence \cite{LepInt04}.  The generating fields then have decomposition and \opes{} given by
\begin{equation} \label{ope:H4}
A(z) = \sum_{n \in \ZZ} A_n z^{-n-1}, \quad A(z) B(w) \sim \frac{\kil{A}{B} k \, \wun}{(z-w)^2} + \frac{\comm{A}{B}(w)}{z-w}, \qquad A,B \in \hfour,
\end{equation}
where $\wun$ denotes the identity field.  If $k \ne 0$, then this can be extended to a \voa{} $\Hfour$ through a variant of the Sugawara construction \cite{FigNon94}:
\begin{equation} \label{eq:DefT}
T(z) = \frac{\normord{E(z) F(z)}}{k} + \frac{\normord{I(z) J(z)}}{k} - \frac{\pd I(z)}{2k} - \frac{\normord{I(z) I(z)}}{2k^2}.
\end{equation}
This is the unique conformal structure that makes the $A(z)$, $A \in \hfour$, into Virasoro primaries of conformal weight $1$.  With $T(z) = \sum_{n \in \ZZ} L_n z^{-n-2}$, we obtain Virasoro modes in (a completion of) the \uea\ $\envalg{\ahfour} / \tideal{K-k\,\wun}$.  The central charge is $c=4$ and the conformal weight of a \hw\ or \lwv\ with $I_0$- and $J_0$-eigenvalues $i$ and $j$ is
\begin{equation} \label{eq:confwt}
	\hwconfwt{i,j} = \frac{i}{k} \brac*{j+\frac{1}{2}-\frac{i}{2k}} \quad \text{or} \quad
	\lwconfwt{i,j} = \frac{i}{k} \brac*{j-\frac{1}{2}-\frac{i}{2k}},
\end{equation}
respectively.  More generally, the conformal weight of a \rhwv\ (this being a weight vector annihilated by $\ahfour^+$) with $I_0$- and $J_0$-eigenvalues $i$ and $j$ is
\begin{equation} \label{eq:confwt'}
\rhwconfwt{i}{h} = \frac{h}{k} + \frac{i}{k} \brac*{\frac{1}{2}-\frac{i}{2k}},
\end{equation}
where $h$ is the eigenvalue of $Q_0 = F_0 E_0 + I_0 J_0$ on the \rhwv.

The conjugation automorphism $\fconjaut$ of $\hfour$ lifts to an automorphism of $\ahfour$ (and thence to $\Hfour$) as follows:
\begin{equation} \label{eq:DefConj}
	\conjaut(E_n) = F_n, \quad \conjaut(I_n) = -I_n, \quad \conjaut(J_n) = -J_n, \quad \conjaut(F_n) = E_n, \quad \conjaut(K) = K.
\end{equation}
The rescale automorphisms $\fscaut_{\alpha}$, $\alpha \in \CC \setminus \set{0}$, and the shift automorphisms $\fshaut_{\beta}$, $\beta \in \CC$, similarly lift:
\begin{equation} \label{eq:DefResc}
	\begin{aligned}
		\scaut_{\alpha}(E_n) &= \alpha^{-1} E_n, & \scaut_{\alpha}(I_n) &= \alpha^{-2} I_n, & \scaut_{\alpha}(J_n) &= J_n, & \scaut_{\alpha}(F_n) &= \alpha^{-1} F_n, & \scaut_{\alpha}(K) &= K, \\
		\shaut_{\beta}(E_n) &= E_n, & \shaut_{\beta}(I_n) &= I_n, & \shaut_{\beta}(J_n) &= J_n - \beta I_n, & \shaut_{\beta}(F_n) &= F_n, & \shaut_{\beta}(K) &= K.
	\end{aligned}
\end{equation}
As in \cref{sec:hfour}, $\shaut_{\beta}$ allows us to tune the parameter $b$ in the bilinear form \eqref{eq:DefKil} to $0$.  Assuming that $k \ne 0$ (so that $\Hfour$ is a \voa), $\scaut_{\alpha}$ lets us tune $a$ to $k^{-1}$.  Because the bilinear form and the level only appear multiplied together in the defining \opes\ \eqref{cr:ah4}, $b=0$ implies that $ak$ is the only other independent parameter.  We will therefore (without loss of generality) assume that $a=k=1$ in what follows.

Aside from these lifts of $\hfour$-automorphisms, there are genuinely new $\ahfour$-automorphisms.  For example, the central element gives us an opportunity to define new shift automorphisms $\shaut'_{\beta}$, $\beta \in \CC$:
\begin{equation} \label{eq:DefShift}
	\shaut'_{\beta}(E_n) = E_n, \quad \shaut'_{\beta}(I_n) = I_n, \quad \shaut'_{\beta}(J_n) = J_n - \beta \delta_{n,0} K, \quad \shaut'_{\beta}(F_n) = F_n, \quad \shaut'_{\beta}(K) = K.
\end{equation}
Finally, we have the spectral flow automorphisms $\sfaut^{\ell}$, $\ell \in \ZZ$, given by
\begin{equation} \label{eq:DefSF}
	\sfaut^{\ell}(E_n) = E_{n-\ell}, \quad \sfaut^{\ell}(I_n) = I_n - \ell \delta_{n,0} K, \quad \sfaut^{\ell}(J_n) = J_n, \quad \sfaut^{\ell}(F_n) = F_{n+\ell}, \quad \sfaut^{\ell}(K) = K.
\end{equation}

Unlike conjugation, these new automorphisms do not define automorphisms of the \voa{} $\Hfour$ because they do not preserve the Virasoro zero mode:
\begin{equation} \label{eq:AutL0}
	\conjaut(L_0) = L_0, \quad \shaut'_{\beta}(L_0) = L_0 - \beta I_0, \quad \sfaut^{\ell}(L_0) = L_0 - \ell J_0.
\end{equation}
However, they do define automorphisms of the underlying vertex algebra and so induce invertible endofunctors on the category of weight modules for $\Hfour$ as in \eqref{eq:DefTwMod}.\footnote{We define a weight module for $\ahfour$ and $\Hfour$ to be a module that decomposes as a direct sum of its weight spaces, where a weight space is defined to be the intersection, assumed to be finite-dimensional, of a simultaneous eigenspace of $I_0$, $J_0$ and $K$ with a generalised eigenspace of $L_0$.}

\section{Highest-weight $\Hfour$-modules} \label{sec:hwH4}

We may construct a weight $\ahfour$-module by inducing a weight $\hfour$-module, equipped with a trivial action of $\ahfour^+$ and letting $K$ act as the identity.  Because $\Hfour$ is universal, all these induced $\ahfour$-modules are also $\Hfour$-modules.

In particular, inducing the Verma $\hfour$-module $\finver{i,j}^+$ results in the (level-$1$) Verma $\ahfour$-module $\ver{i,j}^+$, with respect to the standard Borel subalgebra $\ahfour^+ \oplus \vspn \set{E_0, I_0, J_0, K}$.  Similarly, inducing $\finver{i,j}^-$ gives the Verma module $\ver{i,j}^-$ corresponding to the conjugate Borel subalgebra $\ahfour^+ \oplus \vspn \set{F_0, I_0, J_0, K}$.  We shall denote the irreducible quotient of $\ver{i,j}^{\pm}$ by $\irr{i,j}^{\pm}$.  The (generating) \hwv\ of $\ver{i,j}^+$ and $\irr{i,j}^+$ will be denoted, allowing for a certain abuse of notation, by $\ket{i,j}$.

The following simple result puts strong constraints on the structure of these \hw\ $\ahfour$-modules.  It follows immediately from the formula \eqref{eq:confwt} for the conformal weight of a \hwv.
\begin{lemma} \label{lem:svs}
	Let $\amod{H}$ be a \hw\ $\ahfour$-module generated by a \hwv\ of $I_0$-eigenvalue $i$ and $J_0$-eigenvalue $j$.  Then, any given \sv\ of $\amod{H}$ has $J_0$-eigenvalue $j+m$ and $L_0$-eigenvalue $\hwconfwt{i,j} + im$, for some $m \in \ZZ$.  Clearly, $i$ and $m$ must have the same sign (when nonzero).
\end{lemma}
\noindent One simple consequence is that the conformal weight of any \sv\ of $\ver{0,j}^+$ is $\hwconfwt{0,j} = 0$.  As the same is true for the \hwm\ $\ver{0,j}^+ \bigm/ \ver{0,j-1}^+$, it follows that this quotient is irreducible, identifying it as $\irr{0,j}^+$.  This irreducible is therefore isomorphic to the module $\irr{0,j}$ induced from $\finirr{0,j}$.  The same argument applied to $\ver{0,j}^-$ now establishes that $\irr{0,j}^+ \cong \irr{0,j} \cong \irr{0,j}^-$.  As $\irr{0,0}$ is the vacuum module of the \voa\ $\Hfour$, this proves the following assertion.
\begin{proposition} \label{prop:H4simple}
	The universal \voa\ $\Hfour$ is simple.
\end{proposition}

Another simple consequence of \cref{lem:svs} is that the Verma module $\ver{i,j}^+$ is irreducible if $i$ is irrational.  Moreover, if $0<i<1$, then the constraint on the $J_0$-and $L_0$-eigenvalues means that the only possible \svs\ are multiples of the generating \hwv.  $\ver{i,j}^+$ is thus irreducible for $0<i<1$.  However, $\ver{1,j}^+$ is reducible as taking $m=1$ in \cref{lem:svs} corresponds to the \sv\ obtained by acting with $E_{-1}$ on the \hwv.  In fact, acting with $E_{-1}^m$ results in a \sv\ for all $m \in \NN$.

While this is easy to check directly, another useful way to understand this is to note that $\ver{0,j}^+$ and $\ver{1,j}^+$ are related by the action of the conjugation, shift and spectral flow functors introduced at the end of \cref{sec:Hfour}.  In particular, straightforward computation gives
\begin{equation} \label{eq:twistedhwms}
	\begin{aligned}
		\conjmod{\irr{0,j}} &\cong \irr{0,-j}, & \shmod{\beta}{\irr{0,j}} &\cong \irr{0,j+\beta}, & \sfmod{\pm1}{\irr{0,j}} &\cong \irr{\pm1,j}^{\pm} & &\text{($j\in\CC$)}, \\
		\conjmod{\irr{i,j}^{\pm}} &\cong \irr{-i,-j}^{\mp}, & \shmod{\beta}{\irr{i,j}^{\pm}} &\cong \irr{i,j+\beta}^{\pm}, & \sfmod{\mp1}{\irr{i,j}^{\pm}} &\cong \irr{i\mp1,j}^{\mp} & &\text{($i\ne0,\ j\in\CC$)}, \\
		\conjmod{\ver{i,j}^{\pm}} &\cong \ver{-i,-j}^{\mp}, & \shmod{\beta}{\ver{i,j}^{\pm}} &\cong \ver{i,j+\beta}^{\pm}, & \sfmod{\mp1}{\ver{i,j}^{\pm}} &\cong \ver{i\mp1,j}^{\mp} & &\text{($i,j\in\CC$)},
	\end{aligned}
\end{equation}
where $\beta \in \CC$ and we recall that $\irr{0,j}^{\pm}$ is isomorphic to $\irr{0,j}$.  Note that applying the spectral flow functor $\sfaut^{\ell}$ for other (nonzero) values of $\ell$ results in modules with conformal weights that are unbounded below.  We emphasise that these modules are nevertheless still weight $\Hfour$-modules.

It also follows from \eqref{eq:twistedhwms} that the structure of $\ver{i,j}^+$ is independent of $j$ and is essentially the same as that of
\begin{equation} \label{eq:sflwm}
	\sfaut \conjmod{\ver{i,j}^+} \cong \ver{1-i,-j}^+.
\end{equation}
To understand Verma module structures, it therefore suffices to restrict to $i\ge\frac{1}{2}$.  As we have noted, Verma modules with $\frac{1}{2} \le i \le 1$ are easily analysed.  The case $i>1$ requires more work.
\begin{theorem} \label{thm:bjp}
	\leavevmode
	\begin{enumerate}
		\item For $i \notin \ZZ$, $\ver{i,j}^+$ is irreducible. \label{it:vermairr}
		\item For $i \in \ZZ$, $\ver{i,j}^+$ has a \sv\ $\chi$, unique up to scalar multiples, with $J_0$-eigenvalue $j+1$ or $j-1$, according as to whether $i>0$ or $i\le0$, respectively, and $L_0$-eigenvalue $\hwconfwt{i,j} + \abs{i}$. \label{it:vermasv}
		\item The maximal submodule of $\ver{i,j}^+$ is generated by $\chi$. \label{it:vermamax}
	\end{enumerate}
\end{theorem}
\noindent The (somewhat lengthy) proof of this \lcnamecref{thm:bjp} is deferred to \cref{app:proof}.

The first assertion of \cref{thm:bjp} was proven in \cite[Thm.~3.3]{BaoRep11} using a combinatorial argument split into many cases.  Our proof is likewise combinatorial but shorter, avoiding the need for cases.  \cite[Thm.~3.3]{BaoRep11} also gives necessary conditions for the existence of the \sv\ that we call $\chi$ above.  Our proof demonstrates that a proper subset of these conditions are sufficient and moreover establishes uniqueness.  The third assertion of \cref{thm:bjp} seems to be new.

\section{Relaxed \hw\ $\Hfour$-modules} \label{sec:rhwH4}

We now turn our attention to the \rhwms\ of $\ahfour$ and $\Hfour$.  We first define a \rhwv\ to be a weight vector for $\ahfour$ that is annihilated by $\ahfour^+$.  A \rhw\ $\ahfour$-module is then one that is generated by a single \rhwv.  Every \hw\ $\ahfour$-module is, of course, a \rhwm, but there are many more.  For example, the $\ahfour$-module $\rel{i,[j]}{h}$, with $i\ne0$ and $h \notin i[j]$, induced from the irreducible $\hfour$-module $\finrel{i,[j]}{h}$, is not highest-weight with respect to any Borel subalgebra.  However, it is a \rhw\ $\ahfour$-module.  Note that the minimal conformal weight of $\rel{i,[j]}{h}$ is $\rhwconfwt{i}{h}$, by \eqref{eq:confwt'}.  Vectors of this minimal conformal weight will be referred to as ground states and will be denoted by $\ket{i,j';\rhwconfwt{i}{h}}$, where $i$ and $j'$ are the $I_0$- and $J_0$-eigenvalues, respectively.

Denote by $\sub{i,[j]}{h}$ the sum of all the submodules of $\rel{i,[j]}{h}$ that have zero intersection with the space of ground states.  The quotient
\begin{equation} \label{eq:defE}
	\quo{i,[j]}{h} = \rel{i,[j]}{h}\bigm/\sub{i,[j]}{h}
\end{equation}
is then irreducible because $\sub{i,[j]}{h}$ coincides with the maximal submodule $\maxsub{i,[j]}{h}$ of $\rel{i,[j]}{h}$. This follows from the fact that $\rel{i,[j]}{h}$ is induced from an irreducible $\hfour$-module.  Note that as \rhwms\ are generated by a single weight vector, they have unique maximal submodules.

In case $i=0$ and $h\ne0$, inducing the irreducible $\finrel{0,[j]}{h}$ results in a \rhw\ $\ahfour$-module that we denote by $\rel{0,[j]}{h}$.  These induced modules all turn out to be irreducible, a fact that we prove in \cref{app:i=0proof}.  This completes the list of irreducible \rhwms.
\begin{theorem} \label{prop:classH-mod}
	Every irreducible \rhw\ $\Hfour$-module is isomorphic to one, and only one, of the following modules:
	\begin{itemize}
		\item $\irr{0,j}$, with $j\in\CC$.
		\item $\irr{i,j}^{\pm}$, with $i\in\CC\setminus\set{0}$ and $j\in\CC$.
		\item $\quo{i,[j]}{h}$, with $i\in\CC\setminus\set{0}$, $[j]\in\CC/\ZZ$, $h\in\CC$ and $h/i \notin [j]$.
		\item $\rel{0,[j]}{h}$, with $[j]\in\CC/\ZZ$ and $h\in\CC\setminus\set{0}$.
	\end{itemize}
\end{theorem}

It is straightforward to compute the effect of applying the conjugation and shift functors to each of the $\quo{i,[j]}{h}$ and $\rel{0,[j]}{h}$.  The results are
\begin{equation}
	\begin{aligned}
		\conjmod{\quo{i,[j]}{h}} &\cong \quo{-i,[-j]}{h+i}, & \conjmod{\rel{0,[j]}{h}} &\cong \rel{0,[-j]}{h}, \\
		\shmod{\beta}{\quo{i,[j]}{h}} &\cong \quo{i,[j+\beta]}{h+\beta i}, & \shmod{\beta}{\rel{0,[j]}{h}} &\cong \rel{0,[j+\beta]}{h}.
	\end{aligned}
\end{equation}
On the other hand, applying the spectral flow functor $\sfaut^{\ell}$, with $\ell \ne 0$, to $\quo{i,[j]}{h}$ or $\rel{0,[j]}{h}$ results in an irreducible weight $\Hfour$-module with conformal weights that are unbounded below.

We may also induce reducible $\hfour$-modules, noting that the result will again be reducible.  For example, the induction $\rel{0,j}{0}$ of $\finrel{0,j}{0}$ is reducible.  In this case, the fact that the conformal weight of a \rhwv\ of $\rel{0,j}{0}$ is $0$ shows that every nonzero submodule of $\rel{0,j}{0}$ intersects the space of ground states nontrivially.  In other words, $\sub{0,j}{0} = 0$ for all $j\in\CC$.  However, the maximal submodule $\maxsub{0,j}{0}$ is nonzero --- actually it is isomorphic to $\ver{0,j-1}^+ \oplus \ver{0,j+1}^-$, by \eqref{es:finR} and the exactness of induction.

Consider next the \rhwms\ $\rel{i}{h}^\pm$ that are induced from the reducible dense $\hfour$-modules $\finrel{i}{h}^\pm$. Here, we assume that $i\ne 0$ and recall that the $J_0$-eigenvalues lie in $[h/i]$.  Denote by $\sub{i}{h}^\pm$ the sum of all submodules of $\rel{i}{h}^\pm$ that have zero intersection with the space of ground states and define
\begin{equation} \label{eq:defEpm}
	\quo{i}{h}^\pm = \rel{i}{h}^\pm\bigm/\sub{i}{h}^\pm.
\end{equation}
These quotients are reducible because the maximal submodule $\maxsub{i}{h}^{\pm}$ obviously has nonzero intersection with the space of ground states. However, they are indecomposable and we shall show that they have precisely two composition factors, one \hw\ and one \lw.

First, we again identify the conjugates and shifts of these reducible $\Hfour$-modules:
\begin{equation}
	\begin{aligned}
		\conjmod{\quo{i}{h}^{\pm}} &\cong \quo{-i}{h+i}^{\mp}, &
		\conjmod{\rel{i}{h}^{\pm}} &\cong \rel{-i}{h+i}^{\mp}, &
		\conjmod{\rel{0,j}{0}} &\cong \rel{0,-j}{0}, & \\
		\shmod{\beta}{\quo{i}{h}^{\pm}} &\cong \quo{i}{h+\beta i}^{\pm}, &
		\shmod{\beta}{\rel{i}{h}^{\pm}} &\cong \rel{i}{h+\beta i}^{\pm}, &
		\shmod{\beta}{\rel{0,j}{0}} &\cong \rel{0,j+\beta}{0}.
	\end{aligned}
\end{equation}
As with the irreducible \rhwms, nontrivial spectral flow results in $\Hfour$-modules whose conformal weights are unbounded below.

Next, note that induction being exact means that \eqref{es:finRpm} gives nonsplit exact sequences
\begin{equation} \label{es:Rpm}
	\begin{aligned}
		0 &\lra \ver{i,h/i}^+ \lra \rel{i}{h}^+ \lra \ver{i,h/i+1}^- \lra 0, \\
		0 &\lra \ver{i,h/i+1}^- \lra \rel{i}{h}^- \lra \ver{i,h/i}^+ \lra 0,
	\end{aligned}
\end{equation}
for $i\ne0$.  To deduce analogous sequences for the $\quo{i}{h}^{\pm}$, we need some preparatory lemmas.
\begin{lemma} \label{le:Esubquo}
	If $i\ne0$ and $h \in \CC$, then:
	\begin{enumerate}
		\item The unique irreducible quotient of $\quo{i}{h}^+$ is isomorphic to $\irr{i,h/i+1}^-$.
		\item $\quo{i}{h}^+$ has an irreducible submodule isomorphic to $\irr{i,h/i}^+$.
	\end{enumerate}
\end{lemma}
\begin{proof}
	As $\ver{i,h/i+1}^-$ is a quotient of $\rel{i}{h}^+$, by \eqref{es:Rpm}, so is $\irr{i,h/i+1}^-$. Since \rhwms\ have unique maximal submodules, it follows that $\irr{i,h/i+1}^-$ is the unique simple quotient of $\rel{i}{h}^+$. To establish this for $\quo{i}{h}^+$, note that $\maxsub{i}{h}^+/\sub{i}{h}^+$ is a submodule of $\quo{i}{h}^+$ and that
	\begin{equation}
		\frac{\quo{i}{h}^+}{\maxsub{i}{h}^+/\sub{i}{h}^+} \cong \frac{\rel{i}{h}^+/\sub{i}{h}^+}{\maxsub{i}{h}^+/\sub{i}{h}^+} \cong \frac{\rel{i}{h}^+}{\maxsub{i}{h}^+} \cong \irr{i,h/i+1}^-.
	\end{equation}

	We also know that $\ver{i,h/i}^+$ is a submodule of $\rel{i}{h}^+$, by \eqref{es:Rpm}. Since the maximal submodule $\amod{M}$ of $\ver{i,h/i}^+$ has zero intersection with its space of ground states, this space being isomorphic to the irreducible $\hfour$-module $\finver{i,h/i}^+$, we have $\amod{M} = \ver{i,h/i}^+ \cap \sub{i}{h}^+$. Thus,
	\begin{equation}
		\irr{i,h/i}^+ \cong \frac{\ver{i,h/i}^+}{\amod{M}} = \frac{\ver{i,h/i}^+}{\ver{i,h/i}^+ \cap \sub{i}{h}^+} \cong \frac{\ver{i,h/i}^+ + \sub{i}{h}^+}{\sub{i}{h}^+} \lira \frac{\rel{i}{h}^+}{\sub{i}{h}^+} \cong \quo{i}{h}^+ . \qedhere
	\end{equation}
\end{proof}


Given a weight $\ahfour$-module $\amod{W}$, we shall denote by $\amod{W}(i,j;\Delta)$ its weight space of $I_0$-eigenvalue $i$, $J_0$-eigenvalue $j$ and $L_0$-eigenvalue $\Delta$.
\begin{lemma} \label{le:IndDim}
	Given $i\ne0$ and $h \in \CC$, we have
	\begin{align}
		\sub{i}{h}^+(i,\tfrac{h}{i}+m;\rhwconfwt{i}{h} + n) = \maxsub{i}{h}^+(i,\tfrac{h}{i}+m;\rhwconfwt{i}{h} + n),
	\end{align}
	for all $m > n \ge 0$.
\end{lemma}
\begin{proof}
	We have $\sub{i}{h}^+ \subseteq \maxsub{i}{h}^+$, by definition, so assume that $v \in \maxsub{i}{h}^+(i,\frac{h}{i}+m;\rhwconfwt{i}{h} + n)$ for some $m>n\ge0$. Because each ground state $\ket{i,j;\rhwconfwt{i}{h}}$ with $j > h/i$ generates $\rel{i}{h}^+$, the submodule $\amod{W}_{v}$ generated by $v$ must not contain any ground states with $j > h/i$. Assume that $\amod{W}_{v}$ contains a ground state with $j \le h/i$. By applying \pbw\ basis elements (with mode indices increasing to the right), this implies that $\ket{i,\frac{h}{i}+m-n;\rhwconfwt{i}{h}}$ is also an element of $\amod{W}_{v}$. Since $m-n >0$, this is a contradiction and so $\amod{W}_{v}$ must have zero intersection with the space of ground states. Hence $v \in \sub{i}{h}^+$.
\end{proof}

Combing these lemmas, we obtain the following analogues of \eqref{es:Rpm} for the $\quo{i}{h}^{\pm}$.

\begin{theorem} \label{th:IndecompSES}
	If $i\ne0$ and $h \in \CC$, then we have the following nonsplit short exact sequences:
	\begin{equation} \label{es:Epm}
		\begin{aligned}
			0 &\lra \irr{i,h/i}^+ \lra \quo{i}{h}^+ \lra \irr{i,h/i+1}^- \lra 0, \\
			0 &\lra \irr{i,h/i+1}^- \lra \quo{i}{h}^- \lra \irr{i,h/i}^+ \lra 0.
		\end{aligned}
	\end{equation}
\end{theorem}
\begin{proof}
	From the proof of \cref{le:Esubquo}, we have $\irr{i,h/i}^+ \ira \quo{i}{h}^+$ and $\irr{i,h/i}^+ \cong  \ver{i,h/i}^+ \big/ (\ver{i,h/i}^+ \cap \sub{i}{h}^+)$. It follows that
	\begin{align}
		\frac{\quo{i}{h}^+}{\irr{i,h/i}^+} \cong \quo{i}{h}^+\bigg/\frac{\ver{i,h/i}^+}{\ver{i,h/i}^+ \cap \sub{i}{h}^+} \cong \frac{\rel{i}{h}^+}{\sub{i}{h}^+} \bigg/ \frac{\ver{i,h/i}^+ + \sub{i}{h}^+}{\sub{i}{h}^+} \cong \frac{\rel{i}{h}^+}{\ver{i,h/i}^+ + \sub{i}{h}^+} .
	\end{align}
	From \cref{le:Esubquo} we know that $\irr{i,h/i+1}^-$ is the unique irreducible quotient of $\rel{i}{h}^+$, thus the first short exact sequence will be proved if we can show that $\ver{i,h/i}^+ + \sub{i}{h}^+ = \maxsub{i}{h}^+$.

	It is clear that $\ver{i,h/i}^+ + \sub{i}{h}^+ \subseteq \maxsub{i}{h}^+$, so assume that $v \in \maxsub{i}{h}^+$. Without loss of generality, we may assume that $v$ is a weight vector so suppose that $v \in \maxsub{i}{h}^+(i,j;\hwconfwt{i,j} + n)$ for some $j \in [h/i]$ and some $n\ge0$. By \cref{le:IndDim}, there exists $r\in \NN$ such that $E_0^r$ maps $v$ into $\sub{i}{h}^+$. Furthermore, by the \pbw\ theorem, there exists $s \in \NN$ such that $F_0^s$ maps $v$ into $\ver{i,h/i}^+$. Since each weight space is finite-dimensional, it follows that the image of $v$ in $\maxsub{i}{h}^+\bigm/(\ver{i,h/i}^+ + \sub{i}{h}^+)$ generates a finite-dimensional $\hfour$-module. However, by \cref{le:IZero}, $I_0$ must act as zero on this finite-dimensional $\hfour$-module. Since $i\ne0$, this is only possible if the image of $v$ is $0$. Hence, $\ver{i,h/i}^+ + \sub{i}{h}^+ = \maxsub{i}{h}^+$ as required.

	To get the second short exact sequence, we apply conjugation to the first.  Note that both are necessarily nonsplit because the corresponding short exact sequences \eqref{es:finRpm} of $\hfour$-modules are.
\end{proof}
\noindent Comparing the short exact sequences \eqref{es:Rpm} and \eqref{es:Epm} when $i \notin \ZZ$, \cref{thm:bjp} yields the following assertion.
\begin{corollary} \label{cor:E=R}
	For all $i \notin \ZZ$, we have $\quo{i}{h}^{\pm} \cong \rel{i}{h}^{\pm}$.
\end{corollary}

\section{Characters I} \label{sec:CharsI}

We define the character of an $\Hfour$-module $\amod{W}$ by
\begin{equation}\label{char0}
	\fch{\amod{W}}{\yy,\zz}{\qq} = \traceover{\amod{W}} \yy^{I_0} \zz^{J_0} \qq^{L_0 - c/24} = \sum_{i,j,\Delta \in \CC} \dim \amod{W}(i,j;\Delta) \, \yy^i \zz^j \qq^{\Delta-1/6}.
\end{equation}
Even though the eigenvalue of $I_0$ is fixed on a given indecomposable weight module, the freedom afforded by introducing $\yy$ is expected to be essential to determine modular transformations (which we leave for future work).  As we shall see, the characters \eqref{char0} are most usefully interpreted as formal power series in $\zz$ with coefficients that are holomorphic functions of $\qq$ in the punctured disc $0 < \abs{\qq} < 1$ multiplied by $\yy^i$.

First, note that the effect of the conjugation, shift and spectral flow functors on the character \eqref{char0} is easily determined from \eqref{eq:DefConj}, \eqref{eq:DefShift} and \eqref{eq:DefSF}.
\begin{lemma} \label{lem:twistch}
	We have the following character identities:
	\begin{equation}
		\begin{aligned}
			\fch{\conjmod{\amod{W}}}{\yy,\zz}{\qq} &= \fch{\amod{W}}{\yy^{-1},\zz^{-1}}{\qq}, \\
			\fch{\shmod{\beta}{\amod{W}}}{\yy,\zz}{\qq} &= \zz^{\beta} \fch{\amod{W}}{\yy \qq^{\beta},\zz}{\qq}, \\
			\fch{\sfmod{\ell}{\amod{W}}}{\yy,\zz}{\qq} &= \yy^{\ell} \fch{\amod{W}}{\yy,\zz \qq^{\ell}}{\qq},
		\end{aligned}
	\end{equation}
\end{lemma}
\noindent The characters of the Verma modules, along with their parabolic and relaxed versions, are also easily computed.
\begin{proposition} \label{prop:vermach}
	We have the following character formulae:
	\begin{equation}
		\begin{aligned}
			\fch{\irr{0,j}}{\yy,\zz}{\qq} &= \frac{\yy^i \zz^j \qq^{-1/6}}{\prod_{n=1}^{\infty} (1-\zz^{-1}\qq^n) (1-\qq^n)^2 (1-\zz\qq^n)}
			= \frac{\yy^i \zz^{j+1/2} (1-\zz^{-1})}{\ii \eta(\qq) \fjth{1}{\zz;\qq}}, \\
			\fch{\ver{i,j}^{\pm}}{\yy,\zz}{\qq} &= \frac{\yy^i \zz^j \qq^{\hlwconfwt{i,j} - 1/6}}{\prod_{n=1}^{\infty} (1-\zz^{\mp1}\qq^{n-1}) (1-\qq^n)^2 (1-\zz^{\pm1}\qq^n)}
			= \frac{\yy^i \zz^{j \pm 1/2} \qq^{\hlwconfwt{i,j}}}{\ii \eta(\qq) \fjth{1}{\zz^{\pm1};\qq}}, \\
			\fch{\rel{i,j}{h}}{\yy,\zz}{\qq} &= \frac{\yy^i \zz^j \qq^{\rhwconfwt{i}{h}-1/6}}{\prod_{n=1}^{\infty} (1-\qq^n)^4} \delta(\zz)
			= \frac{\yy^i \zz^j \qq^{\rhwconfwt{i}{h}}}{\eta(\qq)^4} \delta(\zz). \\
		\end{aligned}
	\end{equation}
	Here, $\dfrac{1}{1-x}$ stands for $\displaystyle\sum_{m=0}^{\infty} x^m$ and $\delta(x)$ stands for $\displaystyle\sum_{m=-\infty}^{\infty} x^m$.
\end{proposition}
\noindent These formulae hold irrespective of whether the module is irreducible or not.  In particular, this gives the characters of the $\irr{i,j}^{\pm}$ with $i \notin \ZZ$ (and $i=0$), by \cref{thm:bjp}\ref{it:vermairr}, and the $\quo{i}{h}^{\pm}$ with $i \notin \ZZ$, by \cref{cor:E=R}.

The characters of the remaining $\irr{i,j}^{\pm}$ now follow as consequences of \cref{thm:bjp}\ref{it:vermasv} and \ref{it:vermamax}, whilst those of the remaining $\quo{i}{h}^{\pm}$ follow from \cref{th:IndecompSES,cor:E=R}.
\begin{corollary} \label{cor:hwch}
	We also have the following character formulae, valid for $i \in \ZZ \setminus \set{0}$:
	\begin{equation}
		\fch{\irr{i,j}^{\pm}}{\yy,\zz}{\qq} = (1 - \zz^{\sgn i} q^{\abs{i}}) \fch{\ver{i,j}^{\pm}}{\yy,\zz}{\qq}, \qquad
		\fch{\quo{i}{h}^{\pm}}{\yy,\zz}{\qq} = (1 - q^{\abs{i}}) \fch{\rel{i}{h}^{\pm}}{\yy,\zz}{\qq}.
	\end{equation}
\end{corollary}
\noindent We recall that for $i\ne0$, $\rel{i}{h}^{\pm}$ shares the same character as $\rel{i,h/i}{h}^{\pm}$, given in \cref{prop:vermach}.  These formulae therefore exhaust the characters of all the \rhw\ $\Hfour$-modules introduced in \cref{sec:rhwH4} except for the irreducibles $\quo{i,[j]}{h}$ with $i\ne0$ and $h/i \notin [j]$.  Our last task is thus to compute these remaining characters.

\section{Characters II} \label{sec:Stringiness}

Given a pair $(i,j) \in \CC^2$, we define the string function of an $\ahfour$-module $\amod{W}$ to be the coefficient of $\yy^i \zz^j$ in the character \eqref{char0}:
\begin{equation}
	\fstrfun{i,j}{\amod{W}}{\qq} = \sum_{\Delta \in \CC} \dim \amod{W}(i,j;\Delta) \, \qq^{\Delta-1/6}.
\end{equation}
This is evidently a power series in $\qq$ that converges in the punctured disc $0 < \abs{\qq} < 1$.  Following \cite{KawRel18}, we shall call an $\ahfour$-module \emph{stringy} if its nonzero string functions all coincide.  Examples of stringy modules include the relaxed Verma modules $\rel{i,j}{h}$ and the \rhwms\ $\quo{i}{h}^{\pm}$, $i \in \ZZ \setminus \set{0}$, whose nonzero string functions are determined by \cref{prop:vermach,cor:hwch} to be
\begin{equation} \label{eq:easystrfuns}
	\fstrfun{i,j'}{\rel{i,j}{h}}{\qq} = \frac{\qq^{\rhwconfwt{i}{h}}}{\eta(\qq)^4} \quad \text{and} \quad
	\fstrfun{i,j'}{\quo{i}{h}^{\pm}}{\qq} = \frac{\qq^{\rhwconfwt{i}{h}}(1 - q^{\abs{i}})}{\eta(\qq)^4},
\end{equation}
for $j' \in [j]$ and $j' \in [h/i]$, respectively.

The aim of this \lcnamecref{sec:Stringiness} is to prove that the irreducible \rhw\ $\ahfour$-modules $\quo{i,[j]}{h}$, with $i\ne0$ and $h/i \notin [j]$, are stringy and to determine their string functions. In this way, we shall prove a character formula for these modules.  To establish stringiness, we use the method of \cite{KawRel18}, the key tools being coherent families of \rhw\ $\ahfour$-modules and a generalisation of the Shapovalov form on these families.  The main results of this \lcnamecref{sec:Stringiness} are summarised in the following \lcnamecref{th:MainStringTheorem}.
\begin{theorem} \label{th:MainStringTheorem}
	Take $i\ne0$ and $h/i \notin [j]$.  Then:
	\begin{enumerate}
		\item \label{thp:MST1} $\quo{i,[j]}{h}$ is stringy.
		\item \label{thp:MST2} The nonzero string function of $\quo{i,[j]}{h}$ is given by the limiting string function
		\begin{equation} \label{eq:SimpleStringFunctions}
			\lim_{n \to \infty} \fstrfun{i,h/i+n}{\irr{i,h/i+1}^-}{\qq},
		\end{equation}
		where the limit is understood in the sense of power series in $\qq$ (so taken coefficient-wise).
		\item \label{thp:MST3} For $i \notin \ZZ$, $\quo{i,[j]}{h}$ coincides with $\rel{i,[j]}{h}$ so its character was given in \cref{prop:vermach}.  For $i \in \ZZ \setminus \set{0}$, we instead have
		\begin{equation} \label{eq:irrquoch}
			\fch{\quo{i,[j]}{h}}{\yy,\zz}{\qq} = \frac{\yy^i \zz^j \qq^{\rhwconfwt{i}{h}} (1 - \qq^{\abs{i}})}{\eta(\qq)^4} \delta(\zz).
		\end{equation}
		\item \label{thp:MST4} For $i \in \ZZ \setminus \set{0}$, we have the following nonsplit short exact sequence:
		\begin{equation} \label{es:RRE}
			\dses{\rel{i,[j]}{h+\abs{i}}}{}{\rel{i,[j]}{h}}{}{\quo{i,[j]}{h}}.
		\end{equation}
	\end{enumerate}
\end{theorem}

We will prove this \lcnamecref{th:MainStringTheorem} in a series of steps. We begin with the first assertion, understanding that $i\ne0$ will be assumed for the rest of this \lcnamecref{sec:Stringiness}.
\begin{lemma} \label{le:SimpleStringy}
	The irreducible $\ahfour$-module $\quo{i,[j]}{h}$ is stringy.
\end{lemma}
\begin{proof}
	This follows as in \cite[Thm.~4.7]{KawRel18}, so we only sketch the argument.  From the fact that $E_0$ and $F_0$ act injectively on $\finrel{i,[j]}{h}$, we see that they also act injectively on $\rel{i,[j]}{h}$ and so also on its maximal submodule $\maxsub{i,[j]}{h}$.  This implies that the coefficients of the string functions $\tstrfun{i,j'}{\maxsub{i,[j]}{h}}$ must increase (weakly) under both $j' \mapsto j'+1$ and $j' \mapsto j'-1$.  These string functions are therefore independent of $j' \in [j]$, hence $\maxsub{i,[j]}{h}$ is stringy.  As $\rel{i,[j]}{h}$ is also stringy, the \lcnamecref{le:SimpleStringy} follows.
\end{proof}

To calculate the unique nonzero string function of $\quo{i,[j]}{h}$, we shall show that it coincides with the unique nonzero string function of the reducible $\ahfour$-module $\quo{i}{h}^+$.  Recalling that the former has $h/i \notin [j]$, while the latter has $J_0$-eigenvalues in $[h/i]$, the idea is this can be shown by bringing these modules together in a family so that the $\quo{i,[j]}{h}$ may be directly compared with $\quo{i}{h}^+$.  To facilitate this comparison, we define
\begin{equation} \label{eq:convenience}
	\quo{i,[h/i]}{h} = \quo{i}{h}^+ \quad \text{and} \quad \rel{i,[h/i]}{h} = \rel{i}{h}^+.
\end{equation}
With this convenient notation, we introduce two classes of \emph{coherent families} of \rhw\ $\ahfour$-modules labelled by $i \in \CC \setminus \set{0}$ and $h \in \CC$:
\begin{equation}
	\qcoh{i}{h} = \bigoplus_{[j] \in \CC/\ZZ} \quo{i,[j]}{h} \quad \text{and} \quad
	\rcoh{i}{h} = \bigoplus_{[j] \in \CC/\ZZ} \rel{i,[j]}{h}.
\end{equation}
Coherent families were originally introduced for semisimple Lie algebras in \cite{MatCla00}, where they were used to complete the classification of weight modules with finite-dimensional weight spaces.  However, the idea is quite general.

We shall define an analogue of the Shapovalov form on $\rcoh{i}{h}$, noting that this descends in an obvious fashion to $\qcoh{i}{h}$ by \eqref{eq:defE} and \eqref{eq:defEpm}.  We do this by defining such a form on each $\rel{i,[j]}{h}$ and this requires choosing a cyclic generator.  We shall do this by assuming that $j$ is chosen so that $\ket{i,j;\rhwconfwt{i}{h}}$ is cyclic --- this requires taking $j>h/i$, if $h/i \in [j]$, and otherwise places no constraint on $j$.  A Shapovalov form $\inner{\cdot}{\cdot}_j$ is therefore defined on $\rel{i,[j]}{h}$ by
\begin{equation}
	\begin{gathered}
		\inner[\big]{\ket{i,j;\rhwconfwt{i}{h}}}{\ket{i,j;\rhwconfwt{i}{h}}}_{j} = 1 \\ \text{and} \quad
		\inner[\big]{U\ket{i,j;\rhwconfwt{i}{h}}}{V\ket{i,j;\rhwconfwt{i}{h}}}_{j} = \inner[\big]{\ket{i,j;\rhwconfwt{i}{h}}}{U^\dagger V\ket{i,j;\rhwconfwt{i}{h}}}_{j}, \quad \text{for all}\ U,V \in \envalg{\ahfour}.
	\end{gathered}
\end{equation}
Here, the adjoint of $\envalg{\ahfour}$ is given by the negative of the conjugation automorphism:
\begin{align}
	E_n^\dagger = F_{-n}, \quad F_n^\dagger = E_{-n}, \quad I_n^\dagger = I_{-n}, \quad J_n^\dagger = J_{-n}, \quad K^\dagger = K.
\end{align}
We note that the kernel of the Shapovalov form $\inner{\cdot}{\cdot}_j$ is independent of the choice of $j$ and is exactly the maximal submodule $\maxsub{i,[j]}{h}$.

Extend this to a Shapovalov form $\inner{\cdot}{\cdot}_{i;h}$ on the coherent family $\rcoh{i}{h}$ by taking the direct sum
\begin{align}
	\inner{\cdot}{\cdot}_{i;h} = \bigoplus_{[j] \in \CC/\ZZ} \inner{\cdot}{\cdot}_{j}.
\end{align}
This is equivalent to demanding that distinct eigenspaces of $J_0$ be orthogonal (which is of course consistent with $J_0$ being self-adjoint).

Choose a \pbw\ ordering of $U(\ahfour^- \oplus \CC F_0)$, recalling the triangular decomposition \eqref{triang}, so that mode indices increase to the right.  For each $n \in \NN$, define $P_n$ to be the set of ordered \pbw\ monomials of $U(\ahfour^- \oplus \CC F_0)$ for which the $\ad(J_0)$- and $\ad(L_0)$-eigenvalues are both $-n$.  Since $F_0$ acts injectively on $\rel{i,[j]}{h}$, a basis of the weight space $\rel{i,[j]}{h}(i,j';\rhwconfwt{i}{h}+n)$ is given by the $U\ket{i,j'+n;\rhwconfwt{i}{h}}$, with $U \in P_n$.  Note that $\abs{P_n}$ is clearly finite.

For each $j \in \CC$, we now define the \emph{Shapovalov matrix} for the weight spaces of $\rcoh{i}{h}$.  For $\rcoh{i}{h}(i,j;\rhwconfwt{i}{h}+n)$, it is the $\abs{P_n} \times \abs{P_n}$ matrix
\begin{equation}
	A_{j;n} = \brac[\Big]{\inner[\big]{U\ket{i,j+n;\rhwconfwt{i}{h}}}{V\ket{i,j+n;\rhwconfwt{i}{h}}}_{i;h}}_{U,V \in P_n} .
\end{equation}
The kernel of this matrix is then the weight space $\maxsub{i,[j]}{h}(i,j,\rhwconfwt{i}{h}+n)$.

\begin{lemma} \label{le:ShapovalovRank}
	For each $n\in \NN$, the rank of the Shapovalov matrix $A_{j;n}$ is independent of $j \in \CC$ when the real part of $j$ is sufficiently large.
\end{lemma}
\begin{proof}
	Let $n\in \NN$, $i \in \CC\backslash\{0\}$ and $h \in \CC$.  Then, for all $j \in \CC$, $A_{j;n}$ is a complex matrix and we let $B_{j,n}$ denote its reduced row echelon form over $\CC$.  Now treat $j$ as a formal indeterminate and write $A_n(j)$ for the Shapovalov matrix in this case.  The entries of $A_n(j)$ are thus complex polynomials in $j$ so we may row-reduce over the field $\CC(j)$ of rational functions in $j$ to obtain the reduced row echelon form $B_n(j)$.

	Now, evaluating $B_n(j)$ at a given $j \in \CC$ will result in $B_{j,n}$ for all but finitely many values of $j$, because row reduction gives only finitely many opportunities to divide by zero.  Similarly, each nonzero entry of $B_n(j)$ will evaluate to a nonzero number for all but finitely many values of $j \in \CC$.  As these matrices have finite size, it follows that the number of nonzero rows of $B_{j,n}$ must be the same as the number of nonzero rows of $B_n(j)$ for all but finitely many values of $j \in \CC$.  The rank of $B_{j,n}$ is therefore equal to the rank of $B_n(j)$ for all but finitely many $j \in \CC$ and so the former is independent of $j$ when the real part of $j$ is sufficiently large.
\end{proof}

\begin{proposition} \label{th:StringLimitTheorem}
	For all $i \in \CC \backslash \{0\}$, $[j] \in \CC/\ZZ$ and $h \in \CC$, the limiting string function $\lim_{n \to \infty} \fstrfun{i,j+n}{\quo{i,[j]}{h}}{\qq}$ exists and is independent of $[j]$.
\end{proposition}
\begin{proof}
	For fixed $n \in \NN$, we know from \cref{le:ShapovalovRank} that $\dim \maxsub{i,[j]}{h}(i,j,\rhwconfwt{i}{h}+n)$ is independent of $[j] \in \CC/\ZZ$ when the real part of $j$ is sufficiently large.  On the other hand, $h/i \notin [j]$ gives $\maxsub{i,[j]}{h} = \sub{i,[j]}{h}$ whilst $h/i \in [j]$ gives $\maxsub{i,[j]}{h}(i,j,\rhwconfwt{i}{h}+n) = \sub{i,[j]}{h}(i,j,\rhwconfwt{i}{h}+n)$ when the real part of $j$ is sufficiently large (\cref{le:IndDim}).  It follows that $\dim \sub{i,[j]}{h}(i,j,\rhwconfwt{i}{h} + n)$ is also independent of $[j] \in \CC/\ZZ$ when the real part of $j$ is sufficiently large. Since this is true for all $n$, the limit $\lim_{n \to \infty} \fstrfun{i,j+n}{\quo{i,[j]}{h}}{\qq}$ exists and is independent of $[j] \in \CC/\ZZ$.
\end{proof}

We are now ready to prove the rest of \cref{th:MainStringTheorem}.
\begin{proof}[Proof of \textup{\cref{th:MainStringTheorem}}]
	As part \ref{thp:MST1} was proven in \cref{le:SimpleStringy}, we start with part \ref{thp:MST2}.  This follows by noting that for $[j]=[h/i]$, \cref{th:IndecompSES} gives
	\begin{align}
		\lim_{n \to \infty} \fstrfun{i,h/i+n}{\quo{i,[h/i]}{h}}{\qq}
		&= \lim_{n \to \infty} \fstrfun{i,h/i+n}{\quo{i}{h}^+}{\qq}
		= \lim_{n \to \infty} \left[ \fstrfun{i,h/i+n}{\irr{i,h/i+1}^-}{\qq} + \fstrfun{i,h/i+n}{\irr{i,h/i}^+}{\qq} \right] \\
		&= \lim_{n \to \infty} \fstrfun{i,h/i+n}{\irr{i,h/i+1}^-}{\qq}, \notag
	\end{align}
	because the limiting string function of the \hwm\ $\irr{i,h/i}^+$ is clearly zero.  \cref{th:StringLimitTheorem} then gives
	\begin{align}
		\lim_{n \to \infty} \fstrfun{i,j+n}{\quo{i,[j]}{h}}{\qq} = \lim_{n \to \infty} \fstrfun{i,h/i+n}{\irr{i,h/i+1}^-}{\qq} \quad \text{for all}\ [j] \in \CC/\ZZ,
	\end{align}
	as desired.

	Part \ref{thp:MST3} follows from part \ref{thp:MST1} and the coincidence of the string functions of $\quo{i,[j]}{h}$ and $\quo{i}{h}^+$, the latter being given explicitly in \cref{cor:hwch}.  It follows immediately that
	\begin{equation}
		\ch{\maxsub{i,[j]}{h}} = \ch{\rel{i,[j]}{h}} - \ch{\quo{i,[j]}{h}} = \ch{\rel{i,[j]}{h+\abs{i}}}.
	\end{equation}
	The weight vectors of $\maxsub{i,[j]}{h}$ of conformal weight $\rhwconfwt{i}{h} + \abs{i}$ are therefore \rhwvs.  Since $E_0$ and $F_0$ act injectively on $\maxsub{i,[j]}{h}$ while the \uea\ $\envalg{\ahfour}$ has no zero-divisors, it follows that $\maxsub{i,[j]}{h} \cong \rel{i,[j]}{h+\abs{i}}$, proving part \ref{thp:MST4}.
\end{proof}

We remark that the \rhwvs\ discussed in the proof of part \ref{thp:MST4} of \cref{th:MainStringTheorem} were found in \cite[Thm.~3.4]{BaoRep11}.  Here, we have proven that these vectors generate the maximal submodule $\maxsub{i,[j]}{h}$ of $\rel{i,[j]}{h}$, for $i \in \ZZ \setminus \set{0}$.

\section{Discussion}\label{sec:discus}

We have completed the first step towards a complete analysis of the \nw\ \cft, namely the classification of the irreducible \rhwms\ over the \voa\ $\Hfour$.  We have also computed the characters of all these modules and determined the structure of many reducible but indecomposable $\Hfour$-modules.  A next step would be to compute the modular transformation properties of these characters, following the guidelines set out by the standard module formalism \cite{CreLog13,RidVer14}.  In this framework, we expect the standard $\Hfour$-modules to be the spectral flows of the relaxed Verma modules $\rel{i,j}{h}$, with $i,j,h \in \RR$.

One technical issue is immediately apparent: the characters of the standard modules, given in \cref{prop:vermach}, involve a continuous range of conformal dimensions parametrised linearly by $h \in \RR$.  This is somewhat unusual as one is accustomed to establishing modular S-transforms using quadratic parametrisations.  A similar issue however arises in studies of universal Virasoro and $N=1$ superconformal \vo\ (super)algebra representations, where it was solved by using an alternative parametrisation suggested by the standard free field realisation \cite{MorKac15,CanFusI15}.  Unfortunately, the usual ``Wakimoto-like'' free field realisation of $\Hfour$ \cite{KirStr94} develops a singularity when realising modules with $i=0$, for example the vacuum module.  In particular, the irreducibles $\rel{0,j}{h}$, $h\ne0$, do not seem to be realisable at all.  Nevertheless, one can calculate S-transformation formulae for the standard module characters.  Unfortunately, the $i=0$ singularities result in inconsistent Verlinde computations for the (Grothendieck) fusion coefficients.  We hope to return to this tantalising puzzle in the future.

Another direction to pursue is the construction of logarithmic $\Hfour$-modules, these being reducible but indecomposable representations on which $L_0$ acts nondiagonalisably.  We expect that the simplest examples may be constructed by gluing together the reducible but indecomposable $\Hfour$-modules that we have already analysed.  In particular, we expect that this may be accomplished with appropriate spectral flows of the \rhwms\ $\quo{i}{h}^{\pm}$ of \cref{th:IndecompSES}, in analogy with the known relaxed gluings for affine \voas\ \cite{AdaLat09,RidFus10,CreMod12}.  This might also be possible by instead gluing appropriate quotients of the reducible Verma modules $\ver{i,j}^{\pm}$.  Either way, it would be extremely interesting to investigate whether such logarithmic modules arise naturally as fusion products of the irreducibles classified here.  We also hope to discuss this in the future.

\appendix
\section{A proof of \cref{thm:bjp}} \label{app:proof}

In this \lcnamecref{app:proof}, we prove \cref{thm:bjp} and illustrate it with a simple example.  For convenience, we break the rather lengthy proof into a series of steps.  Let $\ket{i,j}$ denote the (generating) \hwv\ of $\ver{i,j}^+$ and let $\chi$ be a \sv\ of $\ver{i,j}^+$, so $\chi = U \ket{i,j}$ for some $U \in \envalg{\ahfour}$.  From \cref{lem:svs}, there exists $m \in \ZZ$ such that $J_0 \chi = (j+m) \chi$ and $L_0 \chi = (\hwconfwt{i,j} + im) \chi$.  We may (and will) assume that $i \ge \frac{1}{2}$, by \eqref{eq:sflwm}, hence that $m \in \ZZ_{>0}$ ($m=0$ just returns $\chi = \ket{i,j}$).

\subsection*{Step 1.} $\ver{i,j}^+$ is irreducible if $i \notin \ZZ$.

$U$ is thus a linear combination of \pbw-ordered monomials in the $E_{-n}$, $I_{-n}$, $J_{-n}$ and $F_{-n+1}$ with $n\ge1$.  Consider first the condition $0 = I_n \chi = \comm{I_n}{U} \ket{i,j}$, which holds for all $n\ge1$.  As $I_n$ commutes with every mode except $J_{-n}$, the effect of acting with $I_n$ on $\chi$ is to replace $J_{-n}^k$ by $nk J_{-n}^{k-1}$ in each monomial of $U$.  For $k=0$, the monomial is replaced by $0$.  However, making these replacements on monomials with $k>0$ results in a linear combination of linearly independent monomials.  Setting $I_n \chi = 0$ therefore requires that the coefficient of every $k>0$ monomial is zero.  In other words, we may assume that no $J$-modes appear in $U$.

Next, assume that some monomial of $U$ contains a mode $F_{-n}$ and take $n\ge0$ maximal such that $F_{-n}$ does appear (the corresponding monomial comes with nonzero coefficient).  Consider $E_n \chi = 0$, $n\ge0$.  Then, $E_n$ commutes with every mode of $U$ except the $F_{-n+m}$, $0 \le m \le n$, because there are no $J$-modes.  However, $\comm{E_n}{F_{-n+m}} = I_m + n \delta_{m,0}$ commutes with every mode of $U$, again because there are no $J$-modes, and will annihilate $\ket{i,j}$ if $m>0$.  $E_n$ will thus annihilate any monomial of $U \ket{i,j}$ that does not have an $F_{-n}$.  Acting with $E_n$ therefore amounts to replacing $F_{-n}^k$ by $(n+i)k F_{-n}^{k-1}$.  As before, this only gives $0$ if $k=0$ (because $n+i \ge \frac{1}{2}$) and for $k>0$, linear independence of the resulting monomials implies that the coefficient of every monomial with an $F_{-n}$ is $0$.  This contradicts the maximality of $n$ and hence no $F$-modes may appear in $U$.

We now try to repeat the previous argument for $F_n \chi = 0$, $n\ge1$.  So assume that some monomial of $U$ contains an $E_{-n}$ and let $n\ge1$ be maximal such that $E_{-n}$ does appear.  The result of the above argument this time is that $F_n$ replaces $E_{-n}^k$ by $(n-i)k E_{-n}^{k-1}$ throughout.  This is nonzero for all $k>0$, hence no $E$-modes appear in $U$, unless $n=i$.  If $i \in \ZZ_{>0}$, then this fixes the maximal index $n$ such that $E_{-n}$ appears.  However, if $i \notin \ZZ_{>0}$, then $n=i$ is impossible and so only $I$-modes may appear.  But, considering $J_n \chi = 0$, $n\ge1$, rules these out as well.  (Alternatively, a \sv\ with only $I$-modes cannot have the eigenvalues required by \cref{lem:svs}.)  It follows that there is no such singular vector $\chi$, hence $\ver{i,j}^+$ is irreducible, when $i \notin \ZZ_{>0}$.  As \eqref{eq:sflwm} extends this conclusion to all $i \notin \ZZ$, this proves the first assertion of \cref{thm:bjp}.

\subsection*{Step 2.} If $\ver{i,j}^+$ has a \sv\ corresponding to a given $m \in \ZZ_{>0}$, then it is unique up to scalar multiples.

For this, we may assume that $i \in \ZZ_{>0}$.  The arguments above prove that no monomial of $U$ contains a $J_{-n}$ with $n\ge1$, an $F_{-n}$ with $n\ge0$, or an $E_{-n}$ with $n>i$.  We are therefore left with $I$-modes and the $E_{-n}$ for $1 \le n \le i$.  As the $J_0$-eigenvalue of $\chi$ is $j+m$, there must be precisely $m$ $E$-modes in each monomial of $U$.  Moreover, $\comm{L_0}{U} = im$ fixes the sum of the indices of the modes in each monomial.

Let $\partitions{n}$ denote the set of partitions of $n$.  We write $\lambda = [\lambda_1, \lambda_2, \ldots]$ with $\lambda_1 \ge \lambda_2 \ge \cdots$ and shall refer to the $\lambda_k$ as the parts of $\lambda$.  We shall also employ the alternative notation $\lambda = [\lambda_1^{k_1}, \lambda_2^{k_2}, \ldots]$ to indicate that the part $\lambda_1$ has multiplicity $k_1$, and so on.  A subpartition of $\lambda \in \partitions{n}$ is then a partition $\mu$ obtained from $\lambda$ by removing some of its parts.  We denote by $\lambda \setminus \mu$ the partition obtained by removing the parts $\mu_k$ of $\mu$ from those of $\lambda$ (respecting multiplicities).  We shall moreover employ the simplified notation $\lambda \setminus \lambda_k \equiv \lambda \setminus [\lambda_k]$ when removing a single part.

Let $\specparts{i}{m}{\lambda}$ be the set of subpartitions $\mu$ of $\lambda$ having precisely $m$ parts, none of which exceeds $i$.  Then, we may write $\chi$ in the form
\begin{equation} \label{eq:sv}
	\chi = \sum_{\lambda \in \partitions{im}} \sum_{\mu \in \specparts{i}{m}{\lambda}} c(\lambda,\mu) I_{-(\lambda \setminus \mu)} E_{-\mu} \ket{i,j},
\end{equation}
where the $c(\lambda,\mu)$ are unknown coefficients and $A_{-\lambda}$ is shorthand for $A_{-\lambda_1} \cdots A_{-\lambda_{\ell}}$, with $A=E,I,J,F$ and $\lambda$ a partition of precisely $\ell$ parts.  It will be convenient for what follows to set $c(\lambda,\mu) = 0$ if $\mu \notin \specparts{i}{m}{\lambda}$.

Consider now $J_n \chi = 0$, for $n\ge1$.  The result of applying $J_n$ to \eqref{eq:sv} is a sum over $\lambda$ and $\mu$ of terms of the form
\begin{equation}
	\begin{gathered}
		n \mult{n}{\lambda \setminus \mu} c(\lambda,\mu) I_{-(\lambda \setminus \mu \setminus n)} E_{-\mu} \ket{i,j} \\
		\text{and} \quad c(\lambda,\mu) I_{-(\lambda \setminus \mu)} E_{-(\mu \mathrel{-_k} n)} \ket{i,j} \quad \text{($k=1,\dots,m$, $\mu_k > n$)},
	\end{gathered}
\end{equation}
corresponding to commuting $J_n$ with an $I_{-n}$ and an $E_{-\mu_i}$, respectively.  Here, $\mult{n}{\nu}$ denotes the number of parts of $\nu$ equal to $n$ and $\mu \mathrel{-_k} n$ denotes the partition obtained from $\mu$ by subtracting $n$ from $\mu_k$ and reordering parts if necessary.  Linear independence of monomials then gives a constraint for each $\lambda \in \partitions{im}$ and $\mu \in \specparts{i}{m}{\lambda}$:
\begin{equation} \label{eq:mess}
	n \mult{n}{\lambda \setminus \mu} c(\lambda,\mu) + \sum_{k=1}^m c\brac[\big]{\lambda \cup (\mu_k + n) \setminus \mu_k \setminus n, \mu \mathrel{+_k} n} = 0.
\end{equation}
Here, we denote by $\lambda \cup n$ the partition obtained from $\lambda$ by including $n$ as an additional part and $\mu \mathrel{+_k} n$ is the partition obtained from $\lambda$ to adding $n$ to $\lambda_k$ and reordering.  We also understand that a constant of the of the form $c(\lambda' \setminus n,\mu')$ is understood to be $0$ if $n$ is not a part of $\lambda'$.

If $\lambda \ne \mu$, then there exists $n\ge1$ such that $\mult{n}{\lambda \setminus \mu} \ne 0$.  We may therefore solve \eqref{eq:mess} for $c(\lambda,\mu)$ in terms of constants $c(\lambda',\mu')$, where the number of parts of $\lambda'$ is one less than that of $\lambda$.  It follows that the $c(\lambda,\mu)$ are completely determined by the $c(\lambda',\mu')$ in which $\lambda'$ has the minimal possible number of parts.  This obviously occurs when $\lambda' = \mu'$.  However, as $\mu'$ must have $m$ parts and none of its parts may exceed $i$, but their total must be $im$, this forces every part to be $i$.  In other words, every $c(\lambda,\mu)$ is determined by the value of $c([i^m],[i^m])$.  This proves that for each $m \in \ZZ_{>0}$, there is at most one \sv\ $\chi$, up to scalar multiples.

\subsection*{Step 3.} $\ver{i,j}^+$ has a \sv\ for each $m \in \ZZ_{>0}$.

To show that these \svs\ actually exist, it suffices to show existence for $m=1$.  Indeed, if existence holds for $m=1$ then the corresponding \sv\ would generate a submodule of $\ver{i,j}^+$ isomorphic to $\ver{i,j+1}^+$, since $\envalg{\ahfour}$ has no zero-divisors.  As the structure of Verma modules is independent of $j$, $\ver{i,j+1}^+$ also has an $m=1$ \sv\ that generates a submodule isomorphic to $\ver{i,j+2}^+$.  In $\ver{i,j}^+$, this \sv\ corresponds to $m=2$.  This obviously generalises to all $m \in \ZZ_{>0}$.

For $m=1$, the constraint equations \eqref{eq:mess} simplify a little.  Writing $C(\lambda \setminus \lambda_k, \lambda_k) = c(\lambda,\lambda_k)$, $\mu=[\lambda_k]$ and $n=\lambda_{\ell}$, they become
\begin{equation} \label{eq:tidy}
	\lambda_{\ell} \mult{\lambda_{\ell}}{\lambda \setminus \lambda_k} C(\lambda \setminus \lambda_k, \lambda_k) + C(\lambda \setminus \lambda_k \setminus \lambda_{\ell}, \lambda_k + \lambda_{\ell}) = 0.
\end{equation}
As we have seen, these relations show, by recursively removing parts, that $C(\lambda \setminus \lambda_k, \lambda_k)$ is proportional to $C(\varnothing,i)$, where $\varnothing$ denotes the unique partition of $0$, with nonzero proportionality constant.

We first demonstrate that the proportionality constants do not depend on the order in which one removes parts.  If they did, then this would force $C(\varnothing,i) = 0$ and the \sv\ $\chi$ would not exist.  Removing $\lambda_{\ell}$ and then $\lambda_{\ell'}$ from $\lambda$, we have
\begin{equation}
	C(\lambda \setminus \lambda_k, \lambda_k) = \frac{C(\lambda \setminus \lambda_k \setminus \lambda_{\ell} \setminus \lambda_{\ell'}, \lambda_k + \lambda_{\ell} + \lambda_{\ell'})}{\lambda_{\ell} \lambda_{\ell'} \mult{\lambda_{\ell}}{\lambda \setminus \lambda_k} \mult{\lambda_{\ell'}}{\lambda \setminus \lambda_k \setminus \lambda_{\ell}}},
\end{equation}
which is symmetric under $\ell \leftrightarrow \ell'$ if
\begin{equation} \label{eq:mult=mult}
	\mult{\lambda_{\ell}}{\lambda \setminus \lambda_k} \mult{\lambda_{\ell'}}{\lambda \setminus \lambda_k \setminus \lambda_{\ell}} = \mult{\lambda_{\ell'}}{\lambda \setminus \lambda_k} \mult{\lambda_{\ell}}{\lambda \setminus \lambda_k \setminus \lambda_{\ell'}}.
\end{equation}
If $\lambda_{\ell} = \lambda_{\ell'}$, then  \eqref{eq:mult=mult} obviously holds.  But, $\lambda_{\ell} \ne \lambda_{\ell'}$ implies that $\mult{\lambda_{\ell}}{\lambda \setminus \lambda_k} = \mult{\lambda_{\ell}}{\lambda \setminus \lambda_k \setminus \lambda_{\ell'}}$ and $\mult{\lambda_{\ell'}}{\lambda \setminus \lambda_k \setminus \lambda_{\ell}} = \mult{\lambda_{\ell'}}{\lambda \setminus \lambda_k}$, hence \eqref{eq:mult=mult} also holds in this case.  We conclude that $C(\varnothing,i)$ is a free parameter in the solutions of \eqref{eq:tidy}.

To finish the proof of existence of an $m=1$ \sv, we show that a (nonzero) solution of \eqref{eq:tidy} corresponds to a $\chi$ that is annihilated by $\ahfour^+$.  As \eqref{eq:tidy} was deduced from $J_n \chi = 0$, $n\ge1$, and $\ahfour^+$ is generated by $E_0$, $F_1$ and these $J_n$, we only have to check that the solution of \eqref{eq:tidy} we have obtained gives a $\chi = U \ket{i,j}$ satisfying $E_0 \chi = F_1 \chi = 0$.  Since $U$ contains only $I$- and $E$-modes, the former is trivially satisfied.

We therefore consider $-F_1 \chi = 0$ (adding the minus sign for convenience).  Writing the $m=1$ version of \eqref{eq:sv} in the form
\begin{equation} \label{eq:sv'}
	\chi = \sum_{\lambda \in \partitions{i}} \sum_{\lambda_k = 1}^i C(\lambda \setminus \lambda_k,\lambda_k) I_{-(\lambda \setminus \lambda_k)} E_{-\lambda_k} \ket{i,j},
\end{equation}
we see that acting with $-F_1$ amounts to replacing each $E_{-\lambda_k}$ by $I_{-(\lambda_k-1)}$, if $\lambda_k>1$, and by $(i-1)$, if $\lambda_k=1$.  The constraint equations derived from linear independence are therefore
\begin{equation} \label{eq:-f1}
	(i-1) C(\lambda \setminus 1,1) + \sum_{\lambda_{\ell}=1}^{i-1} C(\lambda \setminus \lambda_{\ell} \setminus 1, \lambda_{\ell}+1) = 0.
\end{equation}
Now, substitute $\lambda_k=1$ into \eqref{eq:tidy} to obtain
\begin{equation}
	\lambda_{\ell} \mult{\lambda_{\ell}}{\lambda \setminus 1} C(\lambda \setminus 1, 1) + C(\lambda \setminus \lambda_{\ell} \setminus 1, \lambda_{\ell} + 1) = 0.
\end{equation}
Summing over $\lambda_{\ell}$ from $1$ to $i-1$ then gives \eqref{eq:-f1}, demonstrating that the constraints derived from $-F_1 \chi = 0$ already follow from those derived from $J_n \chi = 0$.  This proves that there is indeed a unique \sv\ $\chi$ corresponding to $m=1$.  Again, \eqref{eq:sflwm} extends this conclusion from $i \in \ZZ_{>0}$ to all $i \in \ZZ$ and so we have established the second assertion of \cref{thm:bjp}.

\medskip

Before tackling the third and last assertion, we detail the existence of \svs\ in the case $i=4$ to illustrate the arguments used above.  In this case, the \sv\ generating the maximal submodule of $\ver{4,j}^+$ has the form $\chi = U \ket{i,j}$ with
\begin{align}
	U = C(\varnothing,4) E_{-4} &+ C([1],3) I_{-1} E_{-3} + C([2],2) I_{-2} E_{-2} + C([1,1],2) I_{-1}^2 E_{-2} \\
	&+ C([3],1) I_{-3} E_{-1} + C([2,1],1) I_{-2} I_{-1} E_{-1} + C([1,1,1],1) I_{-1}^3 E_{-1}. \notag
\end{align}
It is clear that $E_0 \chi = 0$.  From $J_3 \chi = 0$, we obtain
\begin{subequations}
	\begin{equation} \label{eq:j3}
		C(\varnothing,4) + 3 C([3],1) = 0,
	\end{equation}
	whilst $J_2 \chi = 0$ gives
	\begin{equation} \label{eq:j2}
		C(\varnothing,4) + 2 C([2],2) = 0 \quad \text{and} \quad C([1],3) + 2 C([2,1],1) = 0
	\end{equation}
	and $J_1 \chi = 0$ yields instead
	\begin{equation} \label{eq:j1}
		\begin{aligned}
			C(\varnothing,4) + C([1],3) &= 0, & C([1],3) + 2 C([1,1],2) &= 0, \\
			C([2],2) + C([2,1],1) &= 0 & \text{and} \quad C([1,1],2) + 3 C([1,1,1],1) &= 0.
		\end{aligned}
	\end{equation}
\end{subequations}
On the other hand, $F_1 \chi = 0$ instead results in
\begin{equation} \label{eq:f1}
	C(\varnothing,4) + 3 C([3],1) = 0, \quad
	C([1],3) + C([2],2) + 3 C([2,1],1) = 0 \quad \text{and} \quad
	C([1,1],2) + 3 C([1,1,1],1) = 0.
\end{equation}
Note that there is only one part to remove from $[3]$, hence the first equation of \eqref{eq:f1} matches \eqref{eq:j3}.  Similarly, there is only one way to remove a part from $[1,1,1]$, hence the third equation of \eqref{eq:f1} matches the fourth equation of \eqref{eq:j1}.  Finally, there are two ways to remove a part from $[2,1]$, hence the second equation of \eqref{eq:f1} is the sum of the second equation of \eqref{eq:j2} and the third equation of \eqref{eq:j1}.  For completeness, the \sv\ $\chi = U \ket{i,j}$ is explicitly determined (when $C(\varnothing,4) = 1$) by taking
\begin{equation} \label{eq:svi=4}
	U = E_{-4} - I_{-1} E_{-3} - \frac{1}{2} I_{-2} E_{-2} + \frac{1}{2} I_{-1}^2 E_{-2} - \frac{1}{3} I_{-3} E_{-1} + \frac{1}{2} I_{-2} I_{-1} E_{-1} - \frac{1}{6} I_{-1}^3 E_{-1}.
\end{equation}
Every singular vector of $\ver{4,j}^+$ therefore has the form $U^m \ket{i,j}$, for $m \in \NN$.

As an aside, it is actually quite easy to solve the constraints \eqref{eq:tidy} in general.  Writing $\lambda \setminus \lambda_k = [\mu_1^{k_1}, \mu_2^{k_2}, \ldots]$, we find that
\begin{equation} \label{eq:closedformsoln}
	C(\lambda \setminus \lambda_k, \lambda_k) = \frac{(-1)^{k_1+k_2+\cdots}}{k_1! k_2! \cdots \: \mu_1^{k_1} \mu_2^{k_2} \cdots} C(\varnothing,i).
\end{equation}
Substituting into \eqref{eq:sv'} then gives a closed-form formula for $\chi$ when $i \in \ZZ_{>0}$.  \eqref{eq:sflwm} may then be used to obtain a similar formula for $i \in \ZZ_{\le0}$.  It is easy to verify that \eqref{eq:closedformsoln} reproduces \eqref{eq:svi=4} when $i=4$ (and $C(\varnothing,4) = 1$).

\medskip
It remains to prove the final assertion of \cref{thm:bjp}, namely that the maximal submodule of $\ver{i,j}^+$ is generated by the \sv\ $\chi$ corresponding to $m=1$.  Let $M$ denote the submodule of $\ver{i,j}^+$ generated by $\chi$.

\subsection*{Step 4.} The \hwm\ $\ver{i,j}^+ / M$ has no \svs, except multiples of the image $\overline{\ket{i,j}}$ of the \hwv\ $\ket{i,j}$ of $\ver{i,j}^+$, hence it is irreducible.

So, suppose that $\overline{\psi}$ is a \sv\ of $\ver{i,j}^+ / M$.  Without loss of generality, we may choose a representative $\psi \in \ver{i,j}^+$ of $\overline{\psi}$ that is a weight vector.  $\psi$ is then a subsingular vector of $\ver{i,j}^+$ satisfying $\psi \notin M$ and we have $\psi = U \ket{i,j}$, for some $U \in \envalg{\ahfour}$.  Since $\overline{\psi}$ has $J_0$-eigenvalue $j+m$ and $L_0$-eigenvalue $\hwconfwt{i,j} + im$ for some $m \in \NN$, by \cref{lem:svs}, the same is true for $\psi$.

The argument now generalises that used above to analyse the existence of $\chi$.  We start by assuming that $U$ has a $J$-mode and let $n\ge1$ be maximal such that $J_{-n}$ appears.  Choosing an appropriate \pbw-ordering, there exists $k>0$ such that we may write $U = J_{-n}^k V + W$, where $V$ is a linear combination of monomials with no $J_{-n}$-modes and $W$ is a linear combination of monomials with fewer than $k$ $J_{-n}$-modes.  Applying $I_n^k$, we have $I_n^k \overline{\psi} = 0$ and $I_n^k \psi = I_n^k U \ket{i,j} = n^k k! V \ket{i,j}$, from which we conclude that $V \ket{i,j} \in M$ and so $J_{-n}^k V \ket{i,j} \in M$.  In other words, $W \ket{i,j}$ is another representative of $\overline{\psi}$ in which all monomials have fewer than $k$ $J_{-n}$-modes.  Iterating this argument shows that one can choose the representative $\psi$ so that it contains no $J$-modes.

Repeating this argument with $J_{-n}$ replaced by $F_{-n}$, $n\ge0$, and $I_n$ replaced by $E_n$, we see as before (because $n+i\ne0$) that $\psi$ may be chosen so that it also contains no $F$-modes.  With $J_{-n}$ replaced by $E_{-n}$, $n>i$, and $I_n$ replaced by $F_i$, we similarly learn (from $n-i\ne0$) that $\psi$ may be refined to eliminate any $E_{-n}$ modes with $n>i$.  If we can likewise eliminate any $E_{-i}$-modes, then continuing this argument will rule out all $E$-modes.  But, the $J_0$-eigenvalue of $\psi$ is then only consistent with $m=0$.  Thus, $U$ is constant and $\psi$ is proportional to $\ket{i,j}$, as desired.

To eliminate $E_{-i}$-modes, suppose that there is one and let $k\ge1$ be the maximal power with which $E_{-i}$ appears.  Then, we may write $U = V E_{-i}^k + W$, where $V$ is a linear combination of monomials with no $E_{-i}$ and $W$ is a linear combination of monomials with fewer than $k$ $E_{-i}$-modes.  We next recall that the \sv\ corresponding to $m=k$ has the form $\chi^k = U^k \ket{i,j}$ with $U^k = E_{-i}^k + W'$, where $W'$ is a linear combination of monomials in the $E_{-n}$-modes, with $1 \le n \le i$ and with fewer than $k$ $E_{-i}$-modes, and the $I$-modes.  Since $\chi^k \in M$, $V \chi^k \in M$ and so
\begin{equation}
	\overline{\psi} = \overline{\psi - V \chi^k} = \overline{(W - VW') \ket{i,j}},
\end{equation}
that is we can replace the representative $\psi = U \ket{i,j}$ by $U' \ket{i,j}$, noting that $U' = W-VW'$ is a linear combination of monomials each of which has fewer than $k$ $E_{-i}$-modes.  Iterating this construction therefore allows us to find a representative with no $E_{-i}$-modes at all.  As noted above, this proves that $\ver{i,j}^+ / M$ is irreducible for $i\in\ZZ_{>0}$.  Because \eqref{eq:sflwm} easily extends this conclusion to all $i\in\ZZ$, the proof of \cref{thm:bjp} is complete.

\section{An irreducibility proof} \label{app:i=0proof}

This \lcnamecref{app:i=0proof} is devoted to proving that the \rhw\ $\Hfour$-modules $\rel{0,[j]}{h}$ are irreducible, for all $[j] \in \CC / \ZZ$ and $h \ne 0$.  The proof is similar in spirit to that of the irreducibility of the $\irr{0,j}$, itself a corollary of \cref{lem:svs}, but is slightly more involved.  We note that $\rhwconfwt{0}{h} = h$ is the conformal weight of the generating \rhwvs\ $\ket{0,j';h} \in \rel{0,[j]}{h}$, $j' \in [j]$.

Suppose that $\rel{0,[j]}{h}$ is reducible.  Then, there is a \rhwv\ $v \in \rel{0,[j]}{h}$ of conformal weight strictly greater than $h$.  Since $L_0$ acts on \rhwvs\ as $Q_0 = F_0 E_0 + I_0 J_0$ and $I_0$ acts as $0$, it follows that $v$ is an eigenvector of $F_0 E_0$ with eigenvalue greater than $h$.  However, we shall show that the only eigenvalue of $F_0 E_0$ on $\rel{0,[j]}{h}$ is $h$, a contradiction.

Write $v$ as a linear combination of \pbw-ordered monomials of the form
\begin{equation} \label{eq:rpbw}
	F_{-\lambda} E_{-\mu} J_{-\nu} I_{-\rho} \ket{0,j';h},
\end{equation}
where $j' \in [j]$ and $\lambda$, $\mu$, $\nu$ and $\rho$ are partitions (we use the same notational conventions here as in \cref{app:proof}, see \eqref{eq:sv} and the surrounding text).  Acting with $F_0 E_0$ on such a monomial returns $h$ times the monomial plus a number of commutator terms.  These fall into two classes for which we observe the following simple facts:
\begin{itemize}
	\item Commuting either $F_0$ with an $E$-mode or $E_0$ with an $F$-mode increases the number of $I$-modes by $1$.  These $I$-modes commute with every negative mode, so the number of $I$-modes strictly increases when the result is written as a linear combination of monomials \eqref{eq:rpbw}.
	\item Commuting either $F_0$ or $E_0$ with a $J$-mode decreases the number of $J$-modes by $1$.  The result may require further commutation to represent it as a linear combination of monomials \eqref{eq:rpbw}.  However, this will never increase the number of $J$-modes because $J$ is not in $\comm{\hfour}{\hfour}$.  The number of $J$-modes thus strictly decreases in the each summand of the result, when written as a linear combination of monomials \eqref{eq:rpbw}.
\end{itemize}
Noting that any leftover $F_0$ or $E_0$ modes may be commuted to the right and thus change $j'$, this completely accounts for the action of $F_0 E_0$ on the monomials \eqref{eq:rpbw}.

Order these monomials so that the number of $J$-modes weakly increases and, when the number of $J$-modes is the same, so that the number of $I$-modes weakly decreases.  Then, the matrix representing $F_0 E_0$ in the weight space of $\rel{0,[j]}{h}$ containing $v$ is upper-triangular, with $h$ as every diagonal entry.  This is the desired contradiction, hence the proof is complete.

\flushleft
\providecommand{\opp}[2]{\textsf{arXiv:\mbox{#2}/#1}}
\providecommand{\pp}[2]{\textsf{arXiv:#1 [\mbox{#2}]}}

\end{document}